\documentclass{book}
\usepackage[contrib,lang=american]{ems-book} %% change to `american' if you use American English
\usepackage[margin=0pt]{subfig}
\usepackage{bm,colortbl}

\theoremstyle{definition}
\newtheorem{remark}{Remark}[chapter]
\theoremstyle{plain}
\newtheorem{lemma}{Lemma}[chapter]

\newtheorem{theorem}{Theorem}[chapter]
\newtheorem{conjecture}{Conjecture}[chapter]

\begin{document}
\mainmatter

%------
% Insert the title of your paper and (if necessary)
% a short title for the running head.
%------
\title{Periodic striped states in Ising models with dipolar interactions}
\titlemark{Striped states in Ising models with dipolar interactions}

%------
% Insert full names of the authors.
% Add further authors as follows:
%  \emsauthor{2}{}{}
%  \emsauthor{3}{}{}
% etc.
% Abbreviate first names for the running head.
%------
\emsauthor{1}{Davide Fermi}{D.~Fermi}
\emsauthor{2}{Alessandro Giuliani}{A.~Giuliani}

%------
% Use \authormark if the list of authors is too
% long for the running head: \authormark{A.~Doe et al.}
%------

%------
% Add one \emsaffil and one \email for each author.
% NOTE: The address does NOT appear in the paper.
% It will probably be printed in an appendix.
%------
\emsaffil{1}{Universit\`a degli Studi Roma Tre, Dipartimento di Matematica e Fisica, L.go S.L. Murialdo 1, 00146 Roma, Italy \email{fermidavide@gmail.com}}

\emsaffil{2}{Universit\`a degli Studi Roma Tre, Dipartimento di Matematica e Fisica, L.go S.L. Murialdo 1, 00146 Roma, Italy; 
Centro Linceo Interdisciplinare {\textit{Beniamino Segre}}, Accademia Nazionale dei Lincei,
Palazzo Corsini, Via della Lungara 10, 00165 Roma, Italy \email{alessandro.giuliani@uniroma3.it}}

%------
% Add MSC 2020 codes according to www.ams.org/msc/msc2020.html.
% Secondary codes (in square brackets) are optional.
%------
\classification[82B20]{82D40}

%------
% Add a list of keywords.
%------

%------
% Insert your abstract.
%------
\begin{abstract}
We review the problem of determining the ground states of 2D Ising models with nearest neighbor ferromagnetic and dipolar interactions, and prove a new result 
supporting the conjecture that, if the nearest neighbor coupling $J$ is sufficiently large, the ground states are periodic and `striped'.
More precisely, we prove a restricted version of the conjecture, by constructing the minimizers within the variational class of states whose domain walls are arbitrary collections of 
horizontal and/or vertical straight lines. 
\end{abstract}

\makecontribtitle

%------
% INSERT THE BODY OF THE PAPER HERE (except
% acknowledgments, funding info and bibliography)
%------

\section{Brief review of the state-of-the-art and main results}

In this contribution we review the state-of-the-art of a problem that one of us started to investigate with Elliott Lieb and Joel Lebowitz more than 15 years ago, which consists in proving that the 
ground states of a toy model for thin magnetic films with dipolar interactions, in a suitable parameter range, are periodic and `striped', in a sense to 
be clarified soon. We also prove a new result, by characterizing the minimizers of the model within a variational class of states that are generically a-periodic in both coordinate 
directions. Our hope is that the methods employed in its proof will be useful for further progress towards a full characterization of the global minimizers of the model. 

The model of interest is a 2D Ising model whose formal Hamiltonian on $\mathbb Z^2$ reads: 
\begin{equation}\label{eq:0}
\mathcal{H}(\bm{\sigma}) = -\,{J \over 2}\! \sum_{\mbox{{\scriptsize $\begin{matrix} \bm{x}, \bm{y} \!\in\! \mathbb Z^2, \vspace{-0.07cm}\\ |\bm{x} - \bm{y}| \!=\! 1 \end{matrix}$}}} \!\!(\sigma_{\bm{x}}\sigma_{\bm{y}}-1) \,+\, {1 \over 2} \sum_{ \bm{x},\, \bm{y} \in \mathbb Z^2}
{\sigma_{\bm{x}}\sigma_{\bm{y}} - 1 \over |\bm{x} - \bm{y}|^3} \;,
\end{equation}
where $J>0$, $\bm{\sigma} \equiv \{\sigma_{\bm{x}}\}_{\bm{x} \in \mathbb Z^2} \in \{\pm 1\}^{\mathbb Z^2}$ is a generic Ising spin configuration, and it is understood that the 
diagonal terms in the second sum (i.e., those with $\bm{x}=\bm{y}$) must be interpreted as zero. Note that $\mathcal H$ is normalized so that the uniform states with $\sigma_{\bm{x}} \equiv 1$ or $\sigma_{\bm{x}}  \equiv -1$ have zero energy. This model provides an oversimplified description of a thin magnetic film 
with an easy-axis of magnetization orthogonal to the sample; the first term on the right side of \eqref{eq:0} models a short-range ferromagnetic exchange interaction, while the second 
term describes the dipolar interaction among out-of-plane magnetic moments. The two terms compete: while the short-range interaction favors a uniform state, the dipolar term favors a staggered state such that $\sigma_{\bm{x}}=(-1)^{\|\bm{x}\|_1}$ or $\sigma_{\bm{x}}=(-1)^{\|\bm{x}\|_1+1}$ \cite{FILS80}. 
On the basis of numerical evidence and variational calculations, see e.g. \cite{MI1995}, it is believed that, for $J$ large enough, the competition between the two interactions induces the formation of periodic
structures, more precisely of periodic striped states of the form $\bm{\sigma}_{s}(h^*)=(-1)^{\lfloor x_2/h^*\rfloor}$, or translations, or discrete rotations thereof. Here the optimal stripe width $h^*$ 
is the minimizer of $\mathcal E(h)$, the energy per site of the periodic striped state $\bm{\sigma}_{s}(h)$. A proof of the fact that $\bm{\sigma}_{s}(h^*)$ is an infinite volume ground state 
of $\mathcal{H}$ is still open; the problem can be seen as one specific 
instance of the general question of understanding the spontaneous formation of patterns and periodic structures in many  body systems with competing interactions, which is one of the 
big open questions in statistical mechanics and condensed matter (and more: in fluid dynamics, in material science, in evolutionary biology, etc.).

In order to formulate the main questions, review the known results and state the new ones more precisely, let us formulate the problem in a finite box with periodic boundary conditions: 
let $\Lambda_{L} = \mathbb{Z}^2/L\mathbb Z^2$ be a simple cubic 2D torus of integer side $L > 0$, and 
\begin{equation}\label{eq:H}
\mathcal{H}_{L}(\bm{\sigma}) = -\,{J \over 2}\! \sum_{\mbox{{\scriptsize $\begin{matrix} \bm{x}, \bm{y} \!\in\! \Lambda_{L}, \vspace{-0.07cm}\\ |\bm{x} - \bm{y}| \!=\! 1 \end{matrix}$}}} \!\!(\sigma_{\bm{x}}\sigma_{\bm{y}}-1) \,+\, {1 \over 2} \sum_{\bm{x},\, \bm{y} \in \Lambda_{L}}\sum_{\bm{m} \in \mathbb{Z}^2} {\sigma_{\bm{x}}\sigma_{\bm{y}} - 1 \over |\bm{x} - \bm{y} + L \bm{m}|^3} \;.
\end{equation}
Here $\bm{\sigma}$ can be naturally thought of as an infinite volume Ising spin configuration that is $L$-periodic in both coordinate directions. Viceversa, for any $L$-periodic infinite Ising spin 
configuration $\bm{\sigma}$, we let its \emph{energy per site} be denoted by 
\begin{equation}\label{eq:Epersite}
\mathcal{E}(\bm{\sigma}) :={1 \over (nL)^2}\, \mathcal{H}_{nL}(\bm{\sigma})\,,
\end{equation}
which is independent of $n\in\mathbb N$. In particular, if we consider the periodic striped configuration $\bm{\sigma}_{s}(h)$, we let $\mathcal E_s(h):=\mathcal E(\bm{\sigma}_s(h))$. 
It is easy to see that for almost every $J>0$ the minimizer of $\mathcal E_s(h)$ over $\mathbb N$ is unique, and we denote it by $h^*=h^*(J)$; in the complementary exceptional set of $J$, there are two contiguous minimizers, denoted $h^*(J)$ and $h^*(J)+1$. We also denote by $e_0$ the specific ground state energy in the thermodynamic limit: 
\begin{equation}
e_0:=\lim_{L\to\infty} \min_{\bm{\sigma}}\mathcal E(\bm{\sigma}),\end{equation}
where $\min_{\bm{\sigma}}$ is performed over the $L$-periodic infinite spin configurations. 

\begin{conjecture}\label{Conj:1} There exists $J_0>0$ such that, for any $J\ge J_0$ and $L$ an integer multiple of $2h^*$, the only\footnote{More precisely, if $J$ belongs to the exceptional set for 
which $\mathcal E(h)$ has two minimizers and $L$ is an integer multiple of $2h^*(h^*+1)$, then in addition to the stated minimizers there are $4(h^*+1)$ extra ones, namely 
$\bm{\sigma}_s(h^*+1)$, its translations and its discrete rotations.\label{blah}}
minimizers of $\mathcal{H}_L$ are $\bm{\sigma}_s(h^*)$, its translations and its discrete rotations. In particular, $e_0=\mathcal E_s(h^*)$. 
\end{conjecture}

As stated, the conjecture is still open. However, starting from the work \cite{GLL06}, several partial results supporting it have been proved. First of all, from 
\cite{GLL06,GLL07} it follows that, for $J$ large enough and $L$ an integer multiple of $2h^*$, the minimizers of $\mathcal H_L$ in the variational class of quasi-1D states, i.e., 
of states that are translationally invariant in one coordinate direction, are precisely the expected ones. Moreover, in \cite{GLL06}, lower bounds on $e_0$ matching with $\mathcal E_s(h^*)$ 
at dominant order\footnote{A computation shows that $\lim_{J\to\infty}e^{J/2}\mathcal E_{s}(h^*)=c^*<0$. The lower bound derived in \cite{GLL06} has the form $e_0\ge ce^{-J/2}$, with $c<c^*$.} as $J\to\infty$ are derived. The natural analogue of Conjecture \ref{Conj:1} has been proved in \cite{GS16} for a modified model in which the dipolar interaction decaying like  
$1/|\bm{x}-\bm{y}|^3$ is replaced by a faster-decaying polynomial interaction $1/|\bm{x}-\bm{y}|^p$, with $p>4$ (in this case the condition $J\ge J_0$ in the statement of the conjecture must be replaced by $J_c-\epsilon_0\le J<J_c$ for some $\epsilon_0>0$ and $J_c=\sum_{\bm{0}\neq\bm{n}\in\mathbb Z^2}|n_1|/|\bm{n}|^p$). The proof in \cite{GS16} is based on earlier partial results in \cite{GLL11,GLS13,GLS14} and it has been later generalized to a continuum version of the model in dimension $d\ge 2$ and $p\ge d+2-\epsilon$ for some 
$\epsilon=\epsilon(d)>0$ in \cite{DR19,K22}. 

The method of proof of all these papers is based on the use of Block Reflection Positivity (BRP), an extension of the standard Reflection 
Positivity (RP) method first proposed in \cite{GLL06,GLL07}. BRP, compared with standard RP, has the advantage to apply to situations where the Hamiltonian is not RP and even in the presence of boundary conditions different from periodic. The proofs in \cite{DR19,GS16,K22} on the striped periodic nature of the global minimizers of $d\ge 2$ models with polynomial interactions $1/|\bm{x}-\bm{y}|^p$, $p\ge d+2-\epsilon$, additionally require to combine RP with localization estimates into boxes of appropriate size. Further extensions of these ideas 
have been successfully applied to the proof of periodicity of the global minimizers of: 2D models of in-plane spins with dipolar interactions \cite{GLL07}; 2D models of martensitic phase transitions \cite{GMu12}; effective functionals with diffuse interfaces in  the presence of dipolar-like interactions in $d=1$ \cite{GLL09a,BEGM13} and $d\ge 2$ \cite{DR22};
models with competing interactions in a magnetic field or with mass constraint in $d=1$ \cite{GLL09b} and in $d\ge 2$ \cite{DR21}. 

\medskip

In this paper we prove a restricted version of Conjecture \ref{Conj:1} for the model with dipolar interactions $1/|\bm{x}-\bm{y}|^3$ in $d=2$, concerning the periodic striped nature of the minimizers of $\mathcal H_L$ within a variational class of states that are modulated, generically in a-periodic fashion, in both coordinate directions. In order to define this variational class 
more precisely, given an infinite spin configuration $\bm{\sigma}$, let $\Gamma(\bm{\sigma})$ be the corresponding union of Peierls contours, i.e., the union of the unit segments 
dual\footnote{The unit segment dual to an edge $(\bm{x},\bm{y})$ is the one orthogonal to the edge and centered at the center of the edge.}
to the nearest neighbor edges $(\bm{x},\bm{y})$ of $\mathbb  Z^2$ such that $\sigma_{\bm{x}}\neq\sigma_{\bm{y}}$. Moreover, we let $\Omega_L$ be the set of $L$-periodic 
infinite spin configurations $\bm{\sigma}$ such that $\Gamma(\bm{\sigma})$ consists of a union of (horizontal and/or vertical) straight lines. 

\begin{theorem}\label{thm:main}
There exists $J_0>0$ such that, for any $J\ge J_0$ and $L$ an integer multiple of $2h^*$, the only\footnote{Same caveat as in footnote \ref{blah}.}
minimizers of $\mathcal{H}_L$ within $\Omega_L$ are $\bm{\sigma}_s(h^*)$, its translations and its discrete rotations. 
\end{theorem}

The proof of this result, which is presented in the next sections, roughly goes as follows: first of all, we use BRP to prove that the minimizers of $\mathcal{H}_L$ within $\Omega_L$
are necessarily periodic checkerboard states $\bm{\sigma}_c(h_1,h_2)$ consisting of tiles all of sides $h_1,h_2$ in the two coordinate directions, and alternating signs (here $h_1,h_2$ 
are sides to be determined; we allow $h_1$ -- and/or $h_2$ -- to be infinite, in which case 
we identify $\bm{\sigma}_c(\infty,h)$ with $\bm{\sigma}_s(h)$).  Therefore, the problem is reduced to the proof that the minimizers of $\mathcal E(h_1,h_2):=\mathcal E(\bm{\sigma}_c(h_1,h_2))$ over $h_1,h_2\in\mathbb N\cup\{\infty\}$ are $(h_1,h_2)=(\infty,h^*)$ and $(h_1,h_2)=(h^*,\infty)$. 
While in principle such a minimization problem could be solved numerically, or via a computer-assisted proof, we are not aware of any discussion fully addressing this minimization
problem in the literature. The only partial discussion we are aware of in this regard is the one in \cite{MI1995}, where the authors prove numerically that $\min_h\mathcal E(h,h)>\mathcal E_s(h^*)$. It is unclear whether the method of \cite{MI1995} could be extended to prove that $\min_{h_1,h_2}\mathcal E(h_1,h_2)=\mathcal E_s(h^*)$. Even if it were, 
the numerical approach of \cite{MI1995} does not provide any conceptual understanding of why stripes are better than other periodic structures, not even of the square checkerboard
ones, $\bm{\sigma}_c(h,h)$. On the contrary, in the following sections we provide a fully analytic proof that $\min_{h_1,h_2}\mathcal E(h_1,h_2)=\mathcal E_s(h^*)$ and that the unique
minimizers are $(\infty,h^*)$ or $(h^*,\infty)$, by extending ideas introduced 
in \cite{GLL11} and used there to prove that, in the model with polynomial interactions $1/|\bm{x}-\bm{y}|^p$, $p>4$, 
periodic stripes of sufficiently large width $h$ have lower energy than periodic checkerboard with square tiles of side $h$. 
Our proof sheds some light on the reason why it is energetically favorable for the system to form stripes rather than square or rectangular tiles. 
In fact, our strategy consists in exhibiting different `moves' (modifications of the spin configuration $\bm{\sigma}=\bm{\sigma}_c(h_1,h_2)$ in which the restriction $\bm{\sigma}|_A$ to an 
appropriate set $A\subset \mathbb Z^2$ is flipped, while $\bm{\sigma}|_{A^c}$ is kept as is) that strictly decrease the energy, provided that the sides $h_1,h_2$ are in suitable ranges. 
In this way we exclude, for different reasons, that the minimizing sides $h_1^*,h_2^*$ are both too small, or both finite with a too big ratio, etc. We hope that, in perspective, similar 
moves can be used to locally decrease the energy of localized spin configurations, in the same spirit as the local moves that eliminate corners in the proof in \cite{GS16}. 

\section{Striped periodic nature of the constrained minimizers}

In this section we provide the proof of Theorem \ref{thm:main}. We assume $J$ to be sufficiently large and, for simplicity, to belong to the non-exceptional set of values for which 
$h^*=h^*(J)$ is the unique minimizer of $\mathcal E(h)$, the complementary case being left to the reader; see \cite{GLL06} for details about the determination of the minimizers of $\mathcal E(h)$ in the general case. 
For later reference, it is useful to recall here the asymptotic behavior of $h^*$ as $J\to\infty$, which follows from the following asymptotic evaluation of the energy of periodic striped states. 

\begin{lemma}\label{lemma:stripes}
	Asymptotically as $h\to\infty$, we have 	\begin{equation}
	 \mathcal{E}_s(h) = {2 \over h} \Big[ J - 2\log h - \alpha_{s} + \mathcal{O}\big(h^{-1}\big) \Big]\,, \qquad \mbox{for $h \to +\infty$}\,, \label{eq:estripes}
	\end{equation}
	where, denoting by $K_{1}$ the modified Bessel function of imaginary argument\footnote{If $x>0$, $K_1(x):=\frac1{2x}\int_{-\infty}^\infty \frac{e^{ixt}}{(t^2+1)^{3/2}}dt$, see \cite[Eq. 8.432.5]{GR2007}.\label{foot:bessel}}
of order $1$ and by $\gamma = 0.577...$ the Euler--Mascheroni constant,
	\begin{equation}\label{eq:as}
	\alpha_{s} := 2\, \Bigg(1 + \gamma - \log(\pi/2) + 4 \pi  \sum_{j = 1}^{+\infty} \sum_{n = 1}^{+\infty} j\,K_{1}(2\pi j n) \Bigg) = 2.276...\; .
	\end{equation}
\end{lemma}

\begin{remark} An evaluation of the constant $\alpha_s$ in \eqref{eq:estripes}, based on a numerical fit, was performed in \cite{MI1995}, without providing a closed expression for 
$\alpha_s$ (for comparison with \cite{MI1995}, note that a slightly different normalization of the initial Hamiltonian was employed there).
\end{remark}

Note that from the asymptotic formula \eqref{eq:estripes}, it is easy to deduce the asymptotic behavior of $h^*$, which turns 
out to be: 
	\begin{equation}\label{eq:hst}
	 h_{*} :=  c_{*}\, e^{J/2}(1+O(e^{-J/4}))\,, \qquad c_{*} := e^{1 - {\alpha_{s} \over 2}} = 0.871...\;.
	\end{equation}

\begin{proof}[Proof of Lemma \ref{lemma:stripes}]
By direct evaluation we obtain: 
\begin{equation}
\mathcal E_s(h)=\frac{2J}h-\frac2{h}\sum_{n_1\in\mathbb Z} \Bigg(\sum_{n_2=1}^h\frac{n_2}{(n_1^2+n_2^2)^{3/2}}+\sum_{\ell=0}^\infty\ \sum_{n_2 \,=\, (2\ell+1) h+1}^{(2\ell+3)h} {|n_2 - (2\ell+2) h| \over (n_1^2+n_2^2)^{3/2}}
\Bigg).\end{equation}
If we now apply Poisson summation formula to the sum over $n_1 \in \mathbb{Z}$, recalling the definition of $K_1$ in footnote \ref{foot:bessel}, letting $H_{h} := \sum_{n = 1}^h\! 1/n$ be the $h$-th harmonic number, and using the fact that $\int_{-\infty}^\infty (t^2+1)^{-3/2} dt=2$,  we find:
\begin{equation}\label{aligned}\begin{aligned}
 \mathcal{E}_s(h) & = {2 \over h} \Bigg[ J - 2\,H_{h} - 8 \pi\! \sum_{j_1 \,=\, 1}^{+\infty} \sum_{n_2 \,=\, 1}^{h} j_1 K_{1}(2\pi j_1 n_2)\\
& \hspace{1.08cm} -2 \sum_{\ell\,=\,0}^{+\infty}\ \sum_{n_2 \,=\, (2\ell+1) h+1}^{(2\ell+3)h}\! {|n_2 - (2\ell+2) h| \over n_2^2}\\
& \hspace{1.08cm} -8 \pi\! \sum_{j_1 \,=\, 1}^{+\infty}  \sum_{\ell\,=\,0}^{+\infty}\ \sum_{n_2 \,=\, (2\ell+1)h+1}^{(2\ell+3)h}\!\! {|n_2 - (2\ell+2) h| \over n_2}\, j_1 K_{1}(2\pi j_1 n_2)\Bigg]\,.
\end{aligned}\end{equation}
Now, in the limit $h \to +\infty$, we can rewrite $H_{h} = \log h + \gamma + \mathcal{O}\big(h^{-1}\big)$, with $\gamma = 0.577...$ the Euler constant, see \cite[Eq. 0.131]{GR2007}. Moreover, using the fact that 
$0 < \sqrt{z}\,e^{z}\,K_1(z) \leqslant C_1$ for any $z \in [1,+\infty)$ and a suitable $C_1$ \cite[Ch.\,10, Eq.\,10.40.2]{NIST}, one finds that, as $h\to\infty$, 
$\sum_{j_1 \,=\, 1}^{+\infty} \sum_{n_2 \,=\, 1}^{h} j_1 K_{1}(2\pi j_1 n_2)=
\sum_{j_1 \,=\, 1}^{+\infty} \sum_{n_2 \,=\, 1}^{\infty} j_1 K_{1}(2\pi j_1 n_2)+O(h^{-1/2}e^{-2h})$, and that the triple sum in the last line of \eqref{aligned} is $O(h^{-1/2}e^{-2h})$. Finally, 
noting that the double sum in the second line can be rewritten as $-\frac2{h}\sum_{\ell\ge 0}\sum_{\substack{\xi=-1+j/h:\\ j=1,\ldots,2h}}\frac{|\xi|}{(2\ell+2+\xi)^2},$ which is a Riemann sum approximation of 
\begin{equation*}\begin{split}-2\sum_{\ell\ge 0}\int_{-1}^1 d\xi\,\frac{|\xi|}{(2\ell+2+\xi)^2}&=-2\sum_{\ell\,=\,0}^{+\infty} \left[ {2 \over (2 \ell + 2)^2 -1} + \log\left({(2 \ell+ 3)(2\ell+1) \over (2 \ell+ 2)^2}\right) \right]\\
&= -2 +2 \log(\pi/2),\end{split}\end{equation*}
via a straightforward bound on the difference between the Riemann sum and the integral, we find that the double sum in the second line of \eqref{aligned} equals
$-2+2\log(\pi/2) +  \mathcal{O}\big(h^{-1}\big)$. Putting things together, we get the announced result \eqref{eq:estripes}.
\end{proof}

\medskip

Now, consider any state $\bm{\sigma}$ belonging to the variational class $\Omega_L$ under analysis, which was defined right before the statement of Theorem \ref{thm:main}. Let $\Theta(\bm{\sigma})$ be the set of all 
rectangular tiles forming such a state, and let $h_{1}(T)$ and $h_{2}(T)$ denote, respectively, the width and the height of the tile $T \in \Theta(\bm{\sigma})$. A consequence of BRP and, more specifically, of the "Chessboard estimate with open boundary conditions" proved in the appendix of \cite{GLL07}, to be applied here first in the horizontal and then in the vertical direction, is the following lower bound: 
\begin{equation}\label{eq:RP}
\mathcal{H}_{L}(\bm{\sigma}) \geqslant \sum_{T \in \Theta(\bm{\sigma})} |T| \;\mathcal{E}(h_1(T),h_2(T))\,,
\end{equation}
where $|T| = h_1(T)\,h_2(T)$ denotes the area of the single tile $T$, and $\mathcal{E}(h_1,h_2)$ was defined a few lines after the statement of Theorem \ref{thm:main}. 

\begin{remark}
The proof of \eqref{eq:RP} via a two-steps iteration of the chessboard estimate with open boundary conditions crucially requires the set of Peierls contours of $\bm{\sigma}$ to be a union of straight lines: this is 
where we use the structure of the variational class $\Omega_L$ and where the most serious limitation in our main result comes from.
\end{remark}

Our goal is now to show that, for any $(h_1,h_2)\neq(h^*,\infty), (\infty, h^*)$, we have 
\begin{equation}\label{eq:lower}\mathcal{E}(h_1,h_2)>\mathcal E_s(h^*).\end{equation}
If this is the case, taking $L$ to be an integer multiple of $2h^*$ and recalling 
that in this case $\mathcal E_s(h^*)= L^{-2}\,\mathcal H_L(\bm{\sigma}_s(h^*))$, then, in view of \eqref{eq:RP}, we find that $\bm{\sigma}_s(h^*)$ is the unique minimizer of $\mathcal H_L$ in the variational class $\Omega_L$, modulo 
translations and discrete rotations, as desired. The rest of this section is devoted to the proof of \eqref{eq:lower}. Since $\mathcal E(h_1,h_2)=\mathcal E(h_2,h_1)$, with no loss of generality we can assume that 
$$ h_1 \geqslant h_2\,, $$
and we shall do so from now on. 
%For later reference, given any spin configuration $\bm{\sigma}$, we also define
%\begin{equation}\label{eq:Delta}
%\Delta(\bm{\sigma}) := \{ \bm{x} \in \Lambda_{L}\,|\, \sigma_{\bm{x}} = 1\}\,, \qquad 
%\Delta^{c}(\bm{\sigma}) := \{ \bm{x} \in \Lambda_{L}\,|\, \sigma_{\bm{x}} = -1\}\,.
%\end{equation}

\subsection{A priori estimates}

Hereafter we proceed to derive constraints on the sides $h_1,h_2$ of the tiles forming an alleged checkerboard state $\bm{\sigma}_{c}(h_1,h_2)$ of minimal energy. To attain this goal, we implement the following general strategy: we compare and estimate the energies of pairs of configurations which only differ by flipping the spins in suitable regions, identified by the insertion or the removal of parallel domain walls. 

\subsubsection{Excluding thin tiles}

First of all, we expect that, whenever the height $h_2 =\min$ $\{h_1,h_2\}$ of the rectangular tiles forming $\bm{\sigma}_{c}(h_1,h_2)$ is too small, we can lower the energy 
by eliminating two neighboring horizontal domain walls. As a consequence (see the proof of the following lemma), for $h_2$ too small, the energy 
of $\bm{\sigma}_{c}(h_1,h_2)$ is strictly larger than the one of $\bm{\sigma}_{c}(h_1,3h_2)$. 

\begin{lemma}\label{lemma:cI}
For all $J > 0$, if 
\begin{equation}\label{eq:lemmacIa}
h_1 \!\geqslant\! h_2\,, \quad 1 \!\leqslant\! h_2 \!\leqslant \!c_{I}\,e^{J/2},
\end{equation}
where, recalling the definition of $\alpha_s$ in \eqref{eq:as},
\begin{equation}\label{eq:lemmacIb}
c_{I} := {2 \over \pi \sqrt{e}}\, e^{- \alpha_{s}/2} = 0.123\,...\;.
\end{equation}
then 
\begin{equation}
\mathcal{E}(h_1,h_2) > \mathcal{E}(h_1,3h_2).
\end{equation}
\end{lemma}

\begin{proof}[Proof of Lemma \ref{lemma:cI}]
Let $L$ be an integer that is divisible both by $h_1$ and $h_2$.  
Consider the checkerboard state $\bm{\sigma}_{c}(h_1,h_2)$ and denote by $\bm{\sigma}_{I}$ the spin configuration obtained by removing two horizontal domain walls, namely, flipping all spins in a given row (see Fig. \ref{fig:AP1}). 
\begin{figure}[t!]
\subfloat[The grey regions represent the positive tiles of the configuration $\bm{\sigma}_{I}$ described in the proof of Lemma \ref{lemma:cI}.]{\includegraphics[width=.45\textwidth]{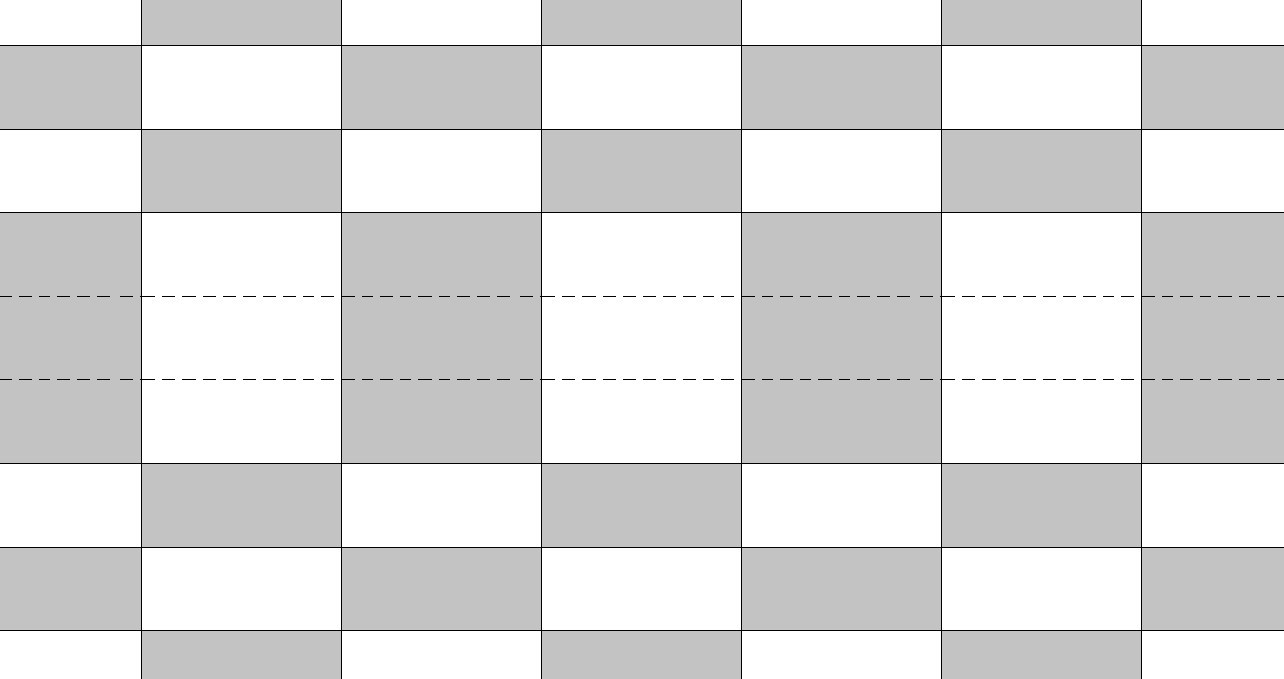}\label{fig:AP1}}
\qquad
\subfloat[The regions $T$ and $\Pi$ considered in the proof of Lemma \ref{lemma:cI}.]{\includegraphics[width=.45\textwidth]{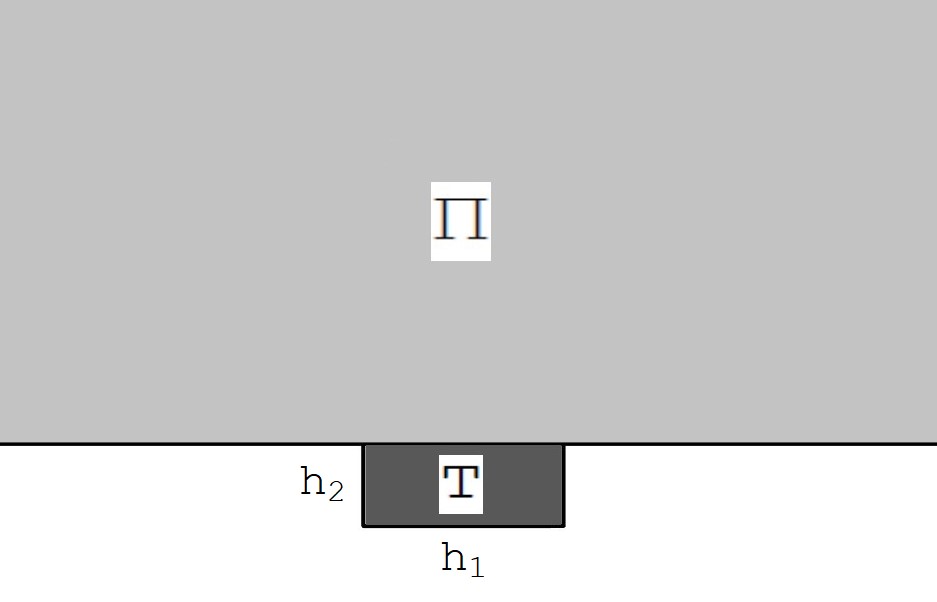}\label{fig:AP1TPi}}
\caption{}
\end{figure}
Let $S^{+}$ (resp. $S^{-}$) be the union of positive (resp. negative) spin tiles of $\bm{\sigma}_c(h_1,h_2)$ contained in $\Lambda_L$ and belonging to the row subject to flipping. 
Moreover, let $\Delta_{e}^{+}$ (resp. $\Delta_{e}^{-}$) be the union of positive (resp. negative) spin tiles of $\bm{\sigma}_c(h_1,h_2)$ contained in $\Lambda_L$ 
which remain unaltered under the flipping. Let also $\Delta\big(\bm{\sigma}_{c}(h_1,h_2)\big) = \Delta_{e}^{+} \cup S^{+}$ and $\Delta^{c}\big(\bm{\sigma}_{c}(h_1,h_2)\big) = \Delta_{e}^{-} \cup S^{-}$ be the union of positive and negative tiles of $\bm{\sigma}_c(h_1,h_2)$ contained in $\Lambda_L$, respectively; similarly, we let $\Delta(\bm{\sigma}_{I}) = \Delta_{e}^{+} \cup S^{-}$ and $\Delta^{c}(\bm{\sigma}_{I}) = \Delta_{e}^{-} \cup S^{+}$  be the union of positive and negative tiles of $\bm{\sigma}_I$, respectively. By direct inspection, using also the 
spin flip symmetry of the energy, we get
\begin{equation*}
\mathcal{H}_{L}\big(\bm{\sigma}_{c}(h_1,h_2)\big) - \mathcal{H}_{L}(\bm{\sigma}_{I})  = 4 J L - 4\, \sum_{\bm{x} \in S^{+}}\Bigg(  \sum_{\bm{y} \in \Delta_{e}^{-}} - \sum_{\bm{y} \in \Delta_{e}^{+}} \Bigg) \sum_{\bm{m} \in \mathbb{Z}^2} {1 \over |\bm{x} - \bm{y} + L \bm{m}|^3}\;.
\end{equation*}
From here, noting that $S^{+}$ consists of $L/(2h_1)$ tiles of size $h_1 \times h_2$, discarding positive contributions, and recalling that $\mathcal E(h_1,h_2)=L^{-2}\mathcal H_L(\bm{\sigma}_c(h_1,h_2))$, we deduce
\begin{equation*}
\mathcal E(h_1,h_2) -\frac1{L^2} \mathcal{H}_{L}(\bm{\sigma}_{I}) > \frac4L\, \Bigg( J - {1 \over h_1} \sum_{\bm{x} \in T}\, \sum_{\bm{y} \in \Pi}\, {1 \over |\bm{x} - \bm{y}|^3}\Bigg) \;,
\end{equation*}
where $T$ is any one of the tiles belonging to $S^{+}$ and $\Pi$ is the half-plane touching $T$ along one of the sides of length $h_1$ (see Fig. \ref{fig:AP1TPi}).
\\
Using Poisson summation formula, and proceeding in a way similar to that described in the proof of Lemma \ref{lemma:stripes}, we infer
\begin{equation}\begin{aligned}\label{eqqo} {1 \over h_1} \sum_{\bm{x} \in T} \sum_{\bm{y} \in \Pi} {1 \over |\bm{x} - \bm{y}|^3} 
& = \sum_{n_1 \in \mathbb{Z}} \sum_{n_2=1}^\infty{\min\{n_2,h_2\} \over (n_1^2+n_2^2)^{3/2}}  \\
& =  \sum_{n_2=1}^\infty\min\{n_2,h_2\} \Bigg(\frac2{n_2^2}+8\pi\sum_{j_1=1}^\infty\frac{j_1}{n_2}K_1(2\pi j_1 n_2)\Bigg)\\
&  \leqslant 2 H_{h_2} + 2h_2 \int_{h_2}^{+\infty} \!{d\eta \over \eta^2} + 8 \pi \sum_{j_1 \,=\,1}^{+\infty} \sum_{n_2\,=\,1}^{+\infty} j_1 K_{1}(2\pi j_1 n_2),\end{aligned}\end{equation}
where in the last line we used the definition of $H_h$, the fact that $\sum_{n=h+1}^\infty n^{-2}<\int_h^\infty x^{-2} dx$, and the positivity of $K_1(x)$ for $x>0$. 
Using the fact that, for all $h\ge 1$, $H_{h} \leqslant \log h + \gamma + {1 \over 2h}$, see \cite[Eq. 0.131]{GR2007}), we find that the last line of \eqref{eqqo}
is bounded from above by $2\log h_2 +2\gamma+2+ 8 \pi \sum_{j,n\ge 1} j K_{1}(2\pi j n)+\frac1{h_2}$. Therefore, recalling the definition \eqref{eq:as} of $\alpha_s$, we find that, for 
$h_2 \geqslant 1$, 
\begin{equation}\label{eqqo.2}
\mathcal E(h_1,h_2) -\frac1{L^2} \mathcal{H}_{L}(\bm{\sigma}_{I}) > \frac4L\, \Big( J -2\log h_2 -\alpha_s-1-2\log(\pi/2)\Big)\;,
\end{equation}
which in turn implies (cf. Eq. \eqref{eq:lemmacIa})
\begin{equation*}
\mathcal E(h_1,h_2)>\frac1{L^2} \mathcal{H}_{L}(\bm{\sigma}_{I})  \,, 
\hspace{0.7cm} \mbox{for all} \hspace{0.7cm}
h_1 \geqslant h_2\,, \quad 1 \leqslant h_2 \leqslant {2 \over \pi \sqrt{e}}\, e^{(J - \alpha_{s})/2} .
\end{equation*}
To conclude, notice that the restriction of $\sigma_{I}$ to $\Lambda_L$ consists of $L/h_1$ tiles of size $h_1 \times 3 h_2$, and of $L^2/(h_1 h_2) - 3L/h_1$ tiles of size $h_1 \times h_2$. Therefore, using \eqref{eq:RP}, we infer, for any $h_1,h_2$ as in Eq. \eqref{eq:lemmacIa},
\begin{equation*}\begin{aligned}
\mathcal E(h_1,h_2)&> \frac1{L^2}\Bigg[\frac{L}{h_1}3h_1h_2\mathcal E(h_1,3h_2)+ \Big(\frac{L^2}{h_1h_2}-\frac{3L}{h_1}\Big)h_1h_2\mathcal E(h_1,h_2)\Bigg]\\
&=\frac{3h_2}L \mathcal E(h_1,3h_2)+\Big(1-\frac{3h_2}L\Big)\mathcal E(h_1,h_2),\end{aligned}\end{equation*}
from which the thesis  readily follows.
\end{proof}

\subsubsection{Excluding thick tiles}

Next, we expect that, if $h_1,h_2$ are both too large, then we can lower the energy of $\bm{\sigma}_{c}(h_1,h_2)$ by creating two extra horizontal domain walls between two 
neighboring pairs thereof. As a consequence (see the proof of the following lemma), if $h_1$ and $h_2$ are both too large, the energy 
of $\bm{\sigma}_{c}(h_1,h_2)$ is strictly larger than the average of the one of $\bm{\sigma}_{c}(h_1,\lfloor h_2/2\rfloor)$ and that of 
$\bm{\sigma}_{c}(h_1,\lceil h_2/2\rceil)$. 

\begin{lemma}\label{lemma:cII}
There exists a (large, compared to 1) positive constant $J_{II}$ such that, for any $0 < \delta \leqslant 1$, $J > J_{II}$ and 
\begin{equation}
 c_{II}(\delta)\,e^{J/2}\! \leqslant h_2 \leqslant \delta\,h_1 , \label{eq:lemmacIIa}
 \end{equation}
 with 
\begin{equation}\label{eq:lemmacIIb}
c_{II}(\delta) := {129 \over 9\pi}\, e^{- (\alpha_s/2) \,+\, 4 \delta} = (1.461\,...\,)\,e^{4\delta} ,
\end{equation}
then 
\begin{equation}
\mathcal{E}(h_1,h_2) > \frac12\Big(\mathcal{E}(h_1,\lfloor h_2/2\rfloor)+\mathcal{E}(h_1,\lceil h_2/2\rceil)\Big)\,.
\end{equation}
\end{lemma}

\begin{proof}[Proof of Lemma \ref{lemma:cII}]
Let $L$ be an integer divisible both by $h_1$ and $h_2$. For simplicity, assume $h_2$ to be even: minor adjustments to the following argument are required if $h_2$ is odd, and these are left to the reader. 
Consider the checkerboard state $\bm{\sigma}_{c}(h_1,h_2)$ and denote by $\bm{\sigma}_{II}$ the configuration with two additional horizontal walls, placed at distance $h_2/2$ from a fixed pre-existent domain wall. Namely, $\bm{\sigma}_{II}$ is obtained starting from $\bm{\sigma}_{c}(h_1,h_2)$ and flipping all spins in a stripe of height $h_2$ placed halfway between two rows, see Fig.\ref{fig:AP2}.
\begin{figure}[t!]
\subfloat[The grey tiles represent the set of positive spins of the configuration $\bm{\sigma}_{II}$ described in the proof of Lemma \ref{lemma:cII}.]{\includegraphics[width=.45\textwidth]{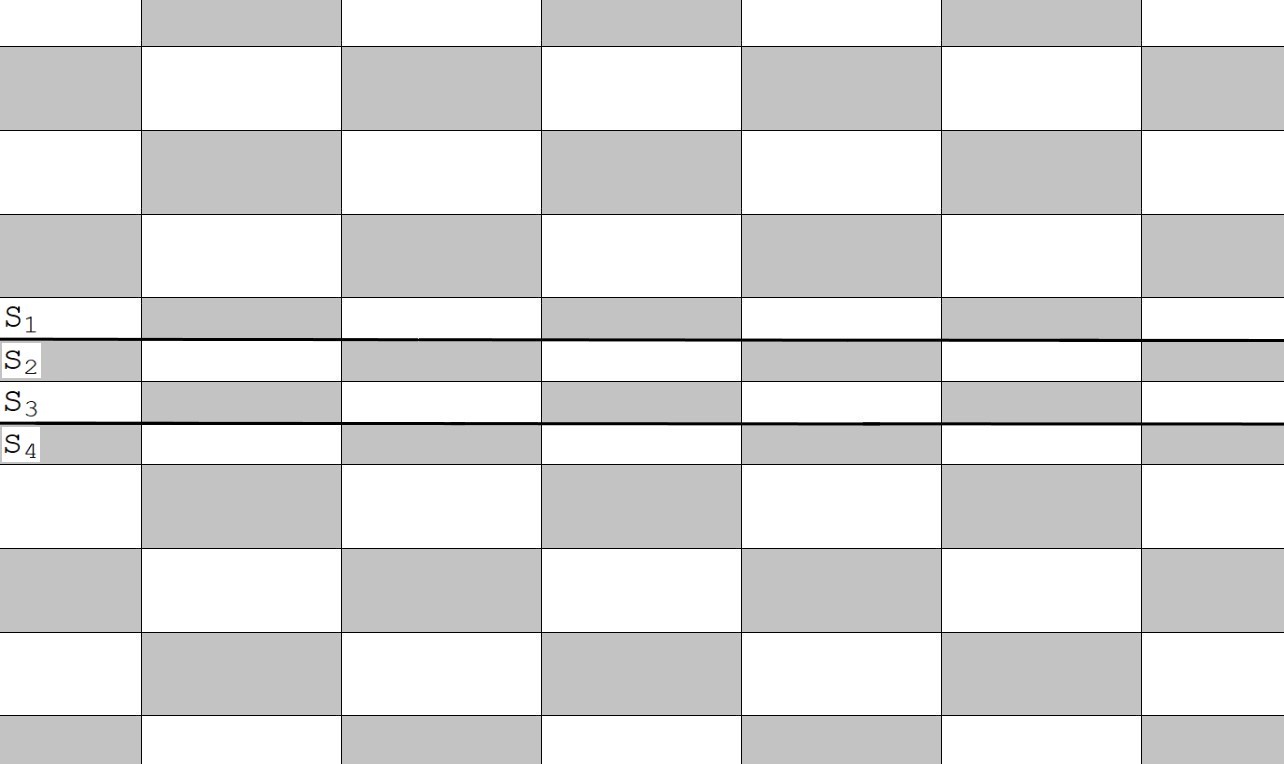}\label{fig:AP2}}
\\
\subfloat[The regions $T$ and $S_{1},S_{4}$ described in the proof of Lemma \ref{lemma:cII}.]{\includegraphics[width=.45\textwidth]{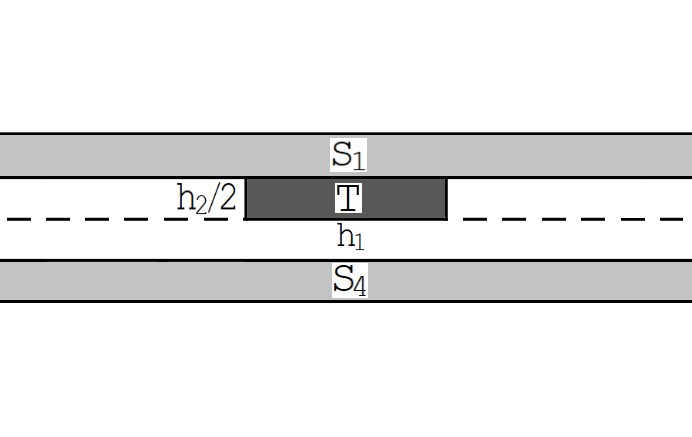}\label{fig:AP2TS1S4}}
\qquad
\subfloat[The regions $T$ and $\Xi$ described in the proof of Lemma \ref{lemma:cII}.]{\includegraphics[width=.45\textwidth]{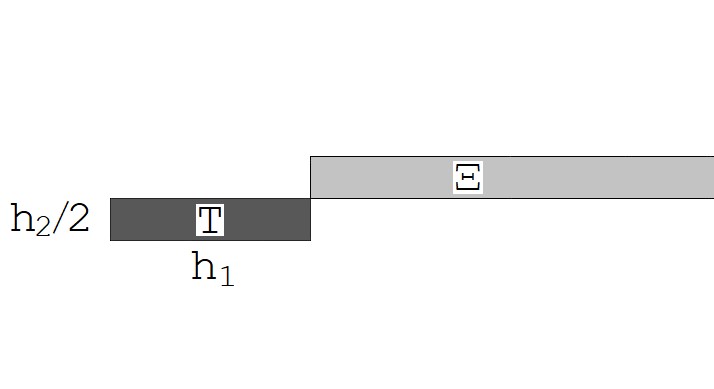}\label{fig:AP2TXi}}
\caption{}
\end{figure}
Consider the four rows of height $h_2/2$  produced by this flipping. We denote by $(S_{i}^{+})_{i = 1,2,3,4}$ (resp. $(S_{i}^{-})_{i = 1,2,3,4}$) the sets of positive (resp. negative) spin sites of $\bm{\sigma}_c(h_1,h_2)$ contained in $\Lambda_L$ and belonging to these four rows, which are numbered in increasing order from top to bottom as in Fig.\ref{fig:AP2}. Moreover, we denote by $\Delta_{e}^{+}$ (resp. $\Delta_{e}^{-}$) the union of the positive (resp. negative) spin tiles of $\bm{\sigma}_c(h_1,h_2)$ contained in $\Lambda_L$ 
which remain unaltered under the flipping. Let also $\Delta\big(\bm{\sigma}_{c}(h_1,h_2)\big) = \Delta_{e}^{+} \cup S_{1}^{+} \cup S_{2}^{+} \cup S_{3}^{+} \cup S_{4}^{+}$ (resp. $\Delta^{c}\big(\bm{\sigma}_{c}(h_1,h_2)\big) = \Delta_{e}^{-} \cup S_{1}^{-} \cup S_{2}^{-} \cup S_{3}^{-} \cup S_{4}^{-}$) be the set of positive (resp. negative) spin sites of $\bm{\sigma}_c(h_1,h_2)$ contained in $\Lambda_L$, and similarly for those of $\bm{\sigma}_{II}$: $\Delta(\bm{\sigma}_{II}) = \Delta_{e}^{+} \cup S_{1}^{+} \cup S_{2}^{-} \cup S_{3}^{-} \cup S_{4}^{+}$ and $\Delta^{c}(\bm{\sigma}_{II}) = \Delta_{e}^{-} \cup S_{1}^{-} \cup S_{2}^{+} \cup S_{3}^{+} \cup S_{4}^{-}$. By direct inspection and the spin flip symmetry of the energy, 
we infer
\begin{equation*}
\begin{split}
& \mathcal{H}_{L}\big(\bm{\sigma}_{c}(h_1,h_2)\big) - \mathcal{H}_{L}(\bm{\sigma}_{II})  \\
& \qquad 
 = - \,4 J L + 8 \sum_{\bm{x} \in S_{2}^{+}} \Big( \sum_{\bm{y} \in \Delta_{e}^{+} \cup S_{1}^{+} \cup S_{4}^{+}} - \sum_{\bm{y} \in \Delta_{e}^{-} \cup S_{1}^{-} \cup S_{4}^{-}} \Big)
 \sum_{\bm{m} \in \mathbb{Z}^2} {1 \over |\bm{x} - \bm{y} + L \bm{m}|^3}\;.\end{split}
\end{equation*}
Notice that $S_{2}^{+}$ consists of $L/(2h_1)$ positive spin half-tiles, each of size $h_1 \times (h_2/2)$. Then, recalling that $\mathcal E(h_1,h_2)=L^{-2}\mathcal H_L(\bm{\sigma}_c(h_1,h_2))$, and proceeding as discussed in the Appendix, we obtain
\begin{align}\label{weare}
\mathcal{E}(h_1,h_2) - L^{-2}\mathcal{H}_{L}(\bm{\sigma}_{II}) > \frac4L\, \bigg[ - J + {1 \over h_1} \sum_{\bm{x} \in T} \bigg( \sum_{\bm{y} \in S_1} - \;2 \sum_{\bm{y} \in S_4} - \;4 \sum_{\bm{y} \in \Xi} \bigg)\, {1 \over |\bm{x} - \bm{y}|^3} \bigg] \,,
\end{align}
where: $T$ is one of the half-tiles in $S_{2}^{+}$; $S_{1} = S_{1}^{+} \cup S_{1}^{-}$ and $S_{4} = S_{4}^{+} \cup S_{4}^{-}$ are the infinite stripes of height $h_2/2$ placed, respectively, at distances $0$ and $h_2/2$ from $T$ (see Fig. \ref{fig:AP2TS1S4}); $\Xi$ is the half-stripe of height $h_2/2$, sharing a vertex with $T$ (see Fig. \ref{fig:AP2TXi}).

We now restrict the attention to $h_1 \geqslant h_2 \geqslant c_{I}\,e^{J/2}$ (cf. Lemma \ref{lemma:cI}) and proceed to examine the case of $J$ large, entailing of course $h_1,h_2$ large.
By proceeding in a way similar to \eqref{eqqo}-\eqref{eqqo.2}, we obtain: 
\begin{equation}\begin{split}
& {1 \over h_1} \sum_{\bm{x} \in T} \sum_{\bm{y} \in S_1} {1 \over |\bm{x} - \bm{y}|^3}
	= \sum_{n_1 \in \mathbb{Z}}  \sum_{n_2\,=\,1}^{h_2}\frac{\min\{n_2, h_2-n_2\}}{(n_1^2+n_2^2)^{3/2}}  \\
& \quad =  \sum_{n_2\,=\,1}^{h_2}\min\{n_2, h_2-n_2\}\Big(\frac2{n_2^2}+\sum_{j\ge 1}8\pi \frac{j}{n_2}K_1(2\pi j n_2)\Big)\\
& \quad \geqslant 2 H_{h_2/2} + 2 \int_{h_2/2}^{h_2}\! {h_2 - x \over x^2} dx + 8 \pi \sum_{j=1}^{\infty}\sum_{n_2\,=\,1}^{+\infty} j_1\,K_{1}(2\pi\,|j_1|\,n_2) +O(h_2^{-1})\\
& \quad = 2 \log h_2 + \alpha_s - 2\log(8/\pi) \,+\, \mathcal{O}\big(h_2^{-1}\big)\,.
\end{split}
\end{equation}
Similarly, using also the fact that $0<K_1(z)\le C_1z^{-1/2}e^{-1}$ for $z\ge 1$ and a suitable $C_1>0$, 
\begin{equation}
\begin{split}
& {1 \over h_1} \sum_{\bm{x} \in T} \sum_{\bm{y} \in S_4} {1 \over |\bm{x} - \bm{y}|^3}
	= \sum_{n_1 \in \mathbb{Z}}\  \sum_{n_2\,=\,h_2/2+1}^{3h_2/2} \frac{\min\{n_2 - h_2/2, 3h_2/2 - n_2\}}{(n_1^2+n_2^2)^{3/2}} \\
& \quad =\sum_{n_2\,=\,h_2/2+1}^{3h_2/2} \min\{n_2 - h_2/2, 3h_2/2 - n_2\}\Bigg(\frac2{n_2^2}+8\pi \sum_{j\ge 1}\frac{j}{n_2}K_1(2\pi j n_2)\Bigg)\\
& \quad \leqslant 2 \int_{h_2/2}^{3h_2/2}\! {\min\{x - h_2/2,3h_2/2 - x\}  \over x^2} dx + \mathcal{O}\big(h_2^{-1}\big) \\
& \quad = 2 \log(4/3) \,+\, \mathcal{O}\big(h_2^{-1}\big)\,,
\end{split}
\end{equation}
and 
\begin{equation}\label{eqqo.3}
{1 \over h_1} \sum_{\bm{x} \in T}\sum_{\bm{y} \in \Xi} {1 \over |\bm{x} - \bm{y}|^3} = {1 \over h_1} \sum_{n_1 \,=\, 1}^{\infty}\min\{ n_1, h_1\}\sum_{n_2\,=\,1}^{h_2} \frac{
\min\{n_2,h_2 - n_2\}}{(n_1^2+n_2^2)^{3/2}}, 
\end{equation}
which, for $h_1$ large, can be thought of as a Riemann sum approximation to 
\begin{equation}\label{int.0}\int_0^\infty dx_1\min\{x_1,1\}\int_0^\zeta dx_2 \frac{\min\{x_2,\zeta-x_2\}}{(x_1^2+x_2^2)^{3/2}},\end{equation}
where $\zeta=h_2/h_1$. An evaluation of this integral and an upper bound on the remainder, i.e., on the difference between \eqref{int.0} and its Riemann sum approximation \eqref{eqqo.3}, leads to the conclusion that, for $h_2$ sufficiently large, \eqref{eqqo.3} is smaller than $2h_2/h_1$. Putting things together, we find
\begin{equation}
\begin{split}
& \mathcal{E}(h_1,h_2) - L^{-2}\mathcal{H}_{L}(\bm{\sigma}_{II})\\
&\qquad  > \frac4L\, \left[-\, J + 2 \log h_2 + \alpha_s - 2\log\left({128 \over 9\pi}\right) - {8h_2 \over h_1} \,+\, \mathcal{O}\big(h_2^{-1}\big) \right], 
\end{split}
\end{equation}
whose right side is strictly positive under the assumptions of the lemma. 

To conclude, notice that the restriction of $\sigma_{II}$ to $\Lambda_L$ consists of $4 L/h_1$ tiles of size $h_1 \times (h_2/2)$, and $L^2/(h_1 h_2) - 2 L/h_1$ tiles of size $h_1 \times h_2$. Therefore, using \eqref{eq:RP}, we infer, for any $h_1,h_2$ as in Eq. \eqref{eq:lemmacIIa},
\begin{align*}
0 & <  \mathcal{E}(h_1,h_2) - L^{-2}\mathcal{H}_{L}(\bm{\sigma}_{II})\\
& \le 
\mathcal{E}(h_1,h_2) - L^{-2}\Bigg[\frac{4 L}{h_1}\frac{h_1h_2}2\mathcal E(h_1,h_2/2)
+\Big(\frac{L^2}{h_1 h_2} - \frac{2 L}{h_1}\Big)h_1h_2\mathcal  E(h_1,h_2)\Bigg]
\end{align*}
which proves the thesis.
\end{proof}

\subsubsection{Excluding long tiles of almost-optimal width}

Finally, we expect that, if $h_2$ is (relatively) close to the optimal width $h^*$, see Eq. \eqref{eq:hst}, and $h_1$ is sufficiently large, 
then we can lower the energy of $\bm{\sigma}_{c}(h_1,h_2)$ by increasing $h_1$. This is proved in the following lemma. 

\begin{lemma}\label{lemma:cIII}
There exists a (large, compared to $1$) positive constant $J_{III}$ such that, for any $0 < \delta \leqslant 1$, $J > J_{III}$, and 
\begin{equation}\label{anot.1}c_I\,e^{J/2} \leqslant h_2 \leqslant \min\!\big\{\delta h_1,\, c_{III}(\delta)\,e^{J/2}\big\}\end{equation} 
with
\begin{equation}\label{eq:lemmacIIIb}
c_{III}(\delta) := {2 \over \pi}\, e^{2 - (\alpha_{s}/2) - \delta/4-\delta^2}= (1.507\,...\,)\,e^{- \delta/4-\delta^2}\, .
\end{equation}
then 
\begin{equation}
\mathcal{E}(h_1,h_2) >\mathcal{E}(3h_1,h_2).\end{equation}
\end{lemma}

\begin{proof}[Proof of Lemma \ref{lemma:cIII}]
Let $L$ be an integer divisible both by $h_1$ and $h_2$.
Consider the usual checkerboard state $\bm{\sigma}_{c}(h_1,h_2)$ and let $\bm{\sigma}_{III}$ be the configuration obtained from $\bm{\sigma}_{c}(h_1,h_2)$
by removing two consecutive vertical domain walls, see Fig.\ref{fig:AP3}.
\begin{figure}[t!]
\subfloat[The grey regions represent the positive tiles of the configuration $\bm{\sigma}_{III}$ described in the proof of Lemma \ref{lemma:cIII}.]{\includegraphics[width=.45\textwidth]{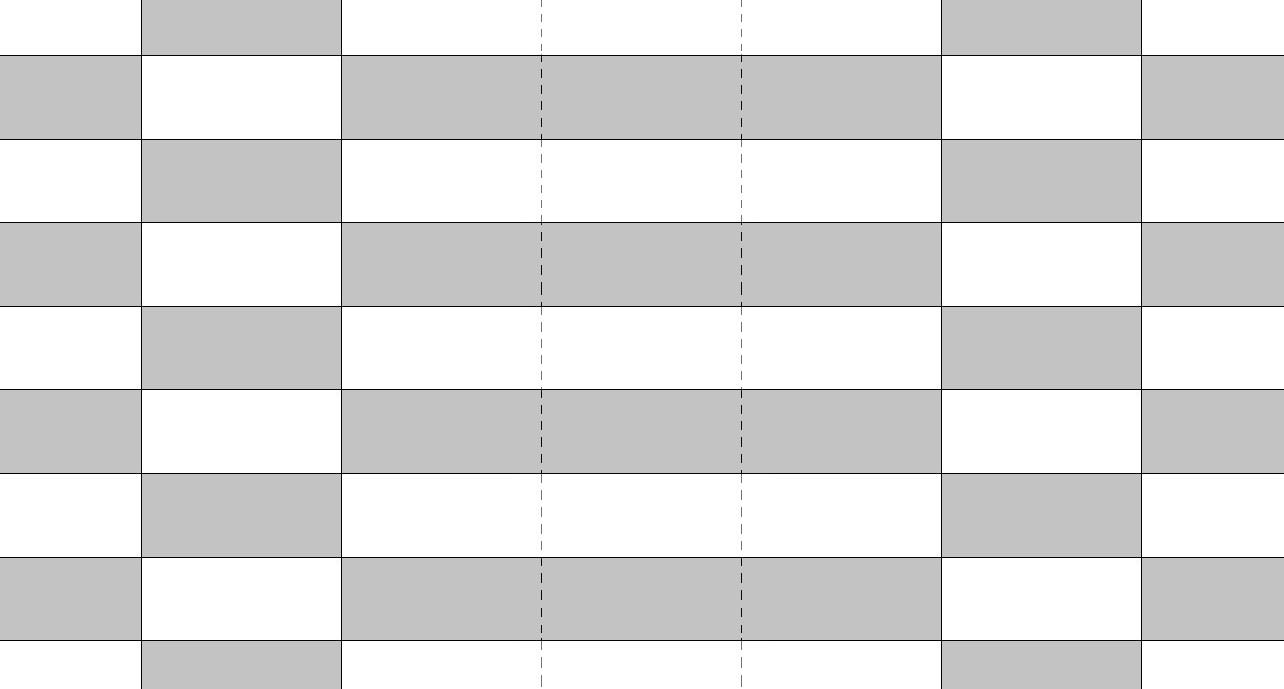}\label{fig:AP3}} 
\qquad
\subfloat[The regions $T_a$ and $\Pi$ described in the proof of Lemma \ref{lemma:cIII}.]{\includegraphics[width=.45\textwidth]{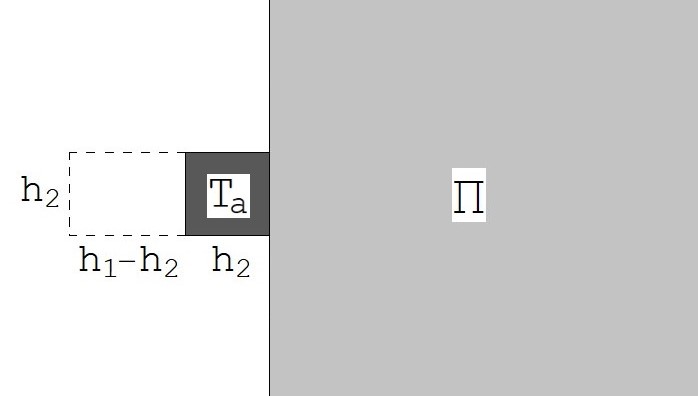}\label{fig:AP3TaPi}}
\\
\subfloat[The regions $T_b$ and $\Xi$ described in the proof of Lemma \ref{lemma:cIII}.]{\includegraphics[width=.45\textwidth]{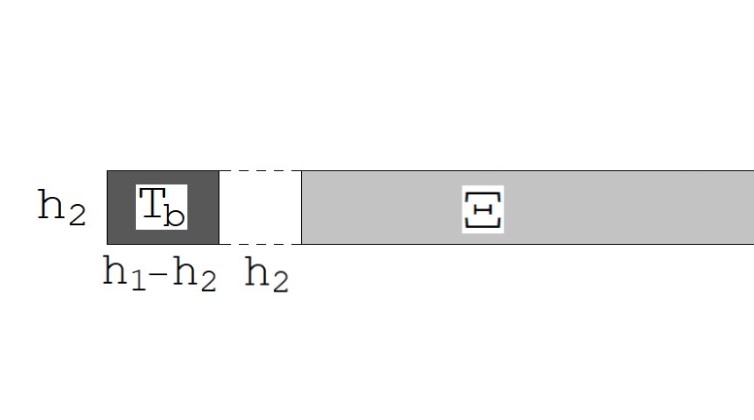}\label{fig:AP3TbXi}}
\qquad
\subfloat[The regions $T_a,P$ and $Q$ described in the proof of Lemma \ref{lemma:cIII}.]{\includegraphics[width=.45\textwidth]{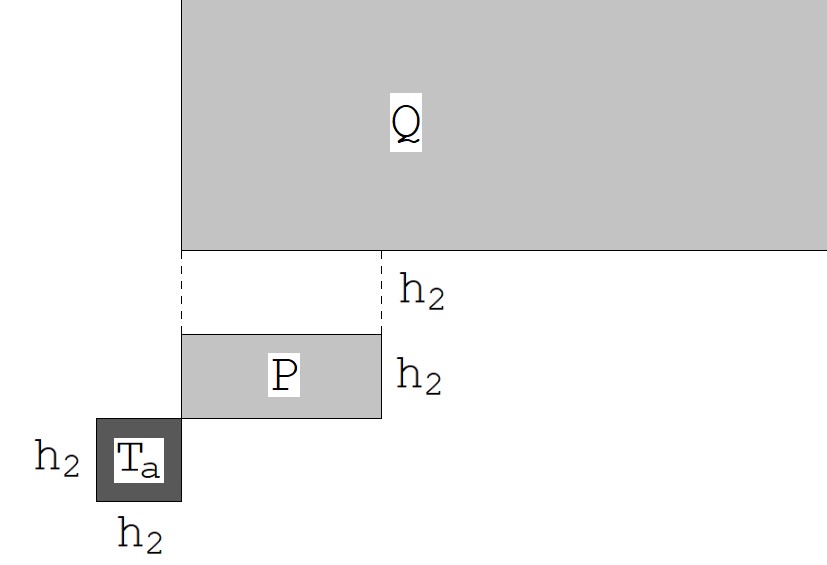}\label{fig:AP3TaPQ}}
\caption{}
\end{figure}
We denote by $U^{+}$ (resp. $U^{-}$) the union of positive (resp. negative) spin tiles of $\bm{\sigma}_{c}(h_1,h_2)$ contained in $\Lambda_L$ and belonging to the column 
subject to flipping. Moreover, let $\Delta_{e}^{+}$ (resp. $\Delta_{e}^{-}$) be the union of positive (resp. negative) spin tiles which remain unaltered under flipping. 
Let also $\Delta\big(\bm{\sigma}_{c}(h_1,h_2)\big) = \Delta_{e}^{+} \cup U^{+}$ (resp. $\Delta^{c}\big(\bm{\sigma}_{c}(h_1,h_2)\big) = \Delta_{e}^{-} \cup U^{-}$)
be the union of positive (resp. negative) spin tiles of $\bm{\sigma}_c(h_1,h_2)$ contained in $\Lambda_L$, and similarly for $\bm{\sigma}_{III}$: 
$\Delta(\bm{\sigma}_{III}) = \Delta_{e}^{+} \cup U^{-}$ and $\Delta^{c}(\bm{\sigma}_{III}) = \Delta_{e}^{-} \cup U^{+}$. In terms of these definitions, we can write
\begin{align*}
& \mathcal{H}_{L}\big(\bm{\sigma}_{c}(h_1,h_2)\big) - \mathcal{H}_{L}(\bm{\sigma}_{III})  \\
& = 4 J L - 4\, \bigg( \sum_{\bm{x} \in U^{+}} \sum_{\bm{y} \in \Delta_{e}^{-}} - \sum_{\bm{x} \in U^{+}} \sum_{\bm{y} \in \Delta_{e}^{+}} \bigg) \sum_{\bm{m} \in \mathbb{Z}^2} {1 \over |\bm{x} - \bm{y} + L \bm{m}|^3} \;.
\end{align*}
Note that $U^{+}$ consists of ${L/2h_2}$ positive spin tiles, each of size $h_1 \times h_2$. Let $T$ be any of these tiles and consider the decomposition $T \equiv T_{a} \cup T_{b}$, where $T_{a}$ is the rightmost square of side $h_2$ contained in $T$, and $T_{b}$ is the complement, i.e., the leftmost rectangle of base $h_1 - h_2$ and height $h_2$. 
Recalling that $\mathcal E(h_1,h_2)=L^{-2}\mathcal{H}_{L}\big(\bm{\sigma}_{c}(h_1,h_2)\big)$, and by proceeding 
similarly to the proof of \eqref{weare}, see Appendix, we deduce
\begin{equation}\label{appp}\begin{split}
& \mathcal{E}(h_1,h_2) - L^{-2}\mathcal{H}_{L}(\bm{\sigma}_{III}) \\
& > \frac4L \,\bigg[J - {1 \over h_2} \bigg(\sum_{\bm{x} \in T_{a}} \sum_{\bm{y} \in \Pi} + \sum_{\bm{x} \in T_{b}} \sum_{\bm{y} \in \Xi} \,-\, 4\! \sum_{\bm{x} \in T_{a}} \sum_{\bm{y} \in P} \,-\, 2\! \sum_{\bm{x} \in T_{a}} \sum_{\bm{y} \in Q}\bigg)\, {1 \over |\bm{x} - \bm{y}|^3} \bigg] \;,\end{split}
\end{equation}
where: $\Pi$ is the half-plane adjacent to $T_{a}$ (see Fig. \ref{fig:AP3TaPi}); $\Xi$ is the half-stripe aligned with $T_{b}$, placed at distance $h_2$ (see Fig. \ref{fig:AP3TbXi}); $P$ is the tile touching $T_{a}$ in one of its vertices, and $Q$ is the quadrant aligned with one of the sides of $T_{a}$ and shifted by $h_2$ in the vertical direction (see Fig. \ref{fig:AP3TaPQ}).

We now restrict the attention to $h_1 \geqslant h_2 \geqslant c_{I}\,e^{J/2}$ (cf. Lemma \ref{lemma:cI}) and proceed to examine the case of large $J$, entailing $h_1,h_2$
large. By proceeding in a way similar to the proof of the previous lemmas, we obtain: 
\begin{align*}
& {1 \over h_2} \sum_{\bm{x} \in T_{a}} \sum_{\bm{y} \in \Pi} {1 \over |\bm{x} - \bm{y}|^3}
	= \sum_{n_1=1}^\infty \min\{n_1,h_2\} \sum_{n_2 \in \mathbb{Z}} {1 \over (n_1^2+ n_2^2)^{3/2}} \\
& \quad = \sum_{n_1=1}^\infty \min\{n_1,h_2\}\Bigg(\frac2{n_1^2}+8\pi\sum_{j=1}^\infty\frac{j}{n_1}K_1(2\pi j n_1)\Bigg)\\
& \quad \leqslant 2 H_{h_2} + 2h_2 \int_{h_2}^{+\infty}\!\! {dx \over x^2} + 8\pi \sum_{n_1 \,=\, 1}^{+\infty} \sum_{j_2 \,=\, 1}^{+\infty} j_2\,K_{1}(2\pi j_2 n_1)
+ \mathcal{O}\big(h_2^{-1}\big) \\
& \quad \leqslant 2\log h_2 + \alpha_{s} + 2\log(\pi/2) + \mathcal{O}\big(h_2^{-1}\big)\,,
\end{align*}
and 
\begin{equation}\label{ert.1} {1 \over h_2} \sum_{\bm{x} \in T_{b}} \sum_{\bm{y} \in \Xi} {1 \over |\bm{x} - \bm{y}|^3}
	= {1 \over h_2} \sum_{n_1 =h_2 + 1}^{\infty}\min\{n_1-h_2, h_1 - h_2\}\sum_{|n_2| \,\leqslant\, h_2}\! {h_2 - |n_2| \over (n_1^2+ n_2^2)^{3/2}} \end{equation}
which, for $h_2$ large, can be thought of as a Riemann sum approximation to 
\begin{equation}\label{int.1}2\int_1^\infty dx_1\min\{x_1-1,Z-1\}\int_0^1 dx_2 \frac{1-x_2}{(x_1^2+x_2^2)^{3/2}},\end{equation}
where $Z=h_1/h_2$. An evaluation of this integral and an upper bound on the remainder, i.e., of the difference between \eqref{int.1} and its Riemann sum approximation \eqref{ert.1} leads to the conclusion that, for $h_2$ sufficiently large, \eqref{ert.1} is smaller than $1/2+h_2/(2h_1)$. Similarly, 
\begin{equation}\label{ert.2}
{1 \over h_2} \sum_{\bm{x} \in T_{a}} \sum_{\bm{y} \in P} {1 \over |\bm{x} - \bm{y}|^3} 
 = {1 \over h_2}\sum_{n_1 \,=\, 1}^{h_1+h_2}\min\{n_1, h_2, h_1 + h_2 - n_1\}\sum_{n_2 \,=\, 1}^{2h_2}\frac{\min\{n_2, 2h_2 - n_2\}}{(n_1^2+ n_2^2)^{3/2}}, \end{equation}
 which, for $h_2$ large, can be thought of as a Riemann sum approximation to 
\begin{equation}\label{int.2}\int_0^{Z+1} dx_1\min\{x_1,1,Z+1-x_1\}\int_0^2 dx_2 \frac{\min\{x_2,2-x_2\}}{(x_1^2+x_2^2)^{3/2}}.\end{equation}
This integral is bounded from below by $\int_0^{Z} dx_1\min\{x_1,1\}\int_0^2 dx_2 \frac{\min\{x_2,2-x_2\}}{(x_1^2+x_2^2)^{3/2}}\equiv (I)-(II)$, where
$$(I)=\int_0^{\infty} dx_1\min\{x_1,1\}\int_0^2 dx_2 \frac{\min\{x_2,2-x_2\}}{(x_1^2+x_2^2)^{3/2}}=0.97229\cdots$$ 
and 
$$(II) = \int_{Z}^\infty dx_1\int_0^2 dx_2 \frac{\min\{x_2,2-x_2\}}{(x_1^2+x_2^2)^{3/2}}\le 
 \int_{Z}^\infty \frac{dx_1}{x_1^3}\int_0^2 dx_2 \min\{x_2,2-x_2\}=\frac1{2Z^2}.$$
Therefore, an upper bound on the difference between \eqref{int.2} and its Riemann sum approximation \eqref{ert.2} leads to the conclusion that, for $h_2$ sufficiently large, \eqref{ert.2} is larger than $0.97-\frac12(h_2/h_1)^2$. Finally, 
\begin{equation}\label{ert.3}
 {1 \over h_2} \sum_{\bm{x} \in T_{a}} \sum_{\bm{y} \in Q} {1 \over |\bm{x} - \bm{y}|^3} 
= {1 \over h_2}\sum_{n_1 \,=\, 1}^{\infty}\min\{n_1,h_2\}\sum_{n_2 \,=\, 2h_2 + 1}^{\infty}\frac{\min\{n_2 - 2h_2, h_2\}}{(n_1^2+ n_2^2)^{3/2}}, \end{equation}
 which, for $h_2$ large, can be thought of as a Riemann sum approximation to 
\begin{equation}\label{int.3}\int_0^{\infty} dx_1\min\{x_1,1\}\int_2^\infty dx_2 \frac{\min\{x_2-2,1\}}{(x_1^2+x_2^2)^{3/2}}=0.36466\cdots\end{equation}
An upper bound on the difference between \eqref{int.3} and its Riemann sum approximation \eqref{ert.3} leads to the conclusion that, for $h_2$ sufficiently large, \eqref{ert.3} is larger than $0.36$. Putting things together, we find that 
\begin{align*}
& \mathcal{E}(h_1,h_2) - L^{-2}\mathcal{H}_{L}(\bm{\sigma}_{III})  \\
& > \frac4{L} \left[J - 2\log h_2 - \alpha_{s} - 2\log(\pi/2) + 4 - {h_2 \over 2h_1}-2\frac{h_2^2}{h_1^2} \right], 
\end{align*}
whose right side is strictly positive under the assumptions of the lemma. 

To conclude, notice that the restriction of $\sigma_{III}$ to $\Lambda_L$ consists of $L/h_2$ tiles of size $3h_1 \times h_2$, and $L^2/(h_1 h_2) - 3 L/h_2$ tiles of size $h_1 \times h_2$. 
Therefore, using \eqref{eq:RP}, we infer, for any $h_1,h_2$ as in \eqref{anot.1},
\begin{align*}
0 & < \mathcal{E}(h_1,h_2) - L^{-2}\mathcal{H}_{L}(\bm{\sigma}_{III}) \\
& \leqslant  \mathcal{E}(h_1,h_2) - L^{-2} \bigg[\frac{L}{h_2} 3h_1h_2\mathcal E(3h_1,h_2)+\Big( \frac{L^2}{h_1 h_2} - 3 \frac{L}{h_2}\Big)h_1h_2\mathcal E(h_1,h_2)\bigg]\,,
\end{align*}
which yields the thesis.
\end{proof}

\begin{figure}[t!]
\subfloat[]{\includegraphics[width=1.\textwidth]{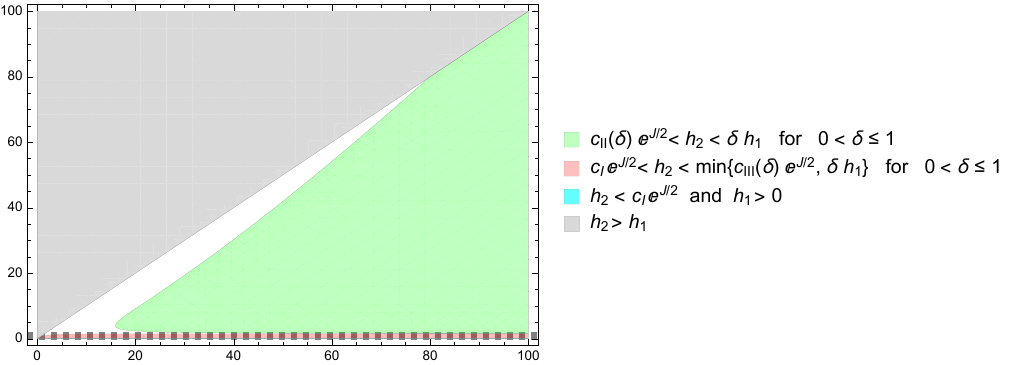}\label{fig:gb1}}
\\
\subfloat[]{\includegraphics[width=.45\textwidth]{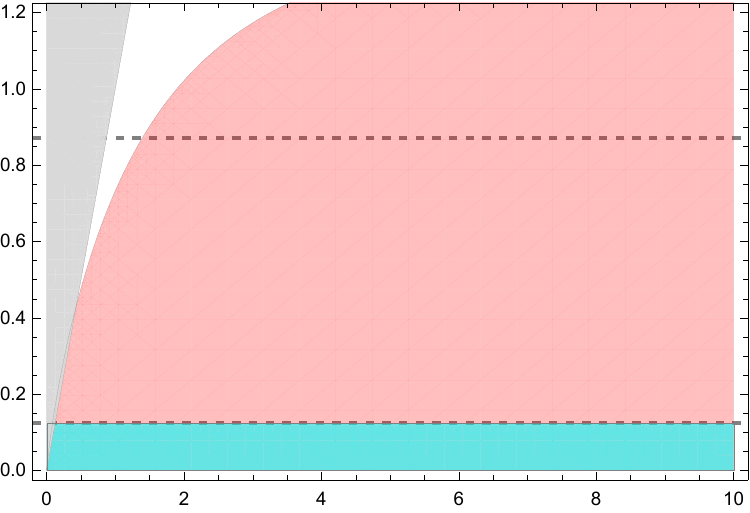}\label{fig:gb2}}
\qquad
\subfloat[]{\includegraphics[width=.45\textwidth]{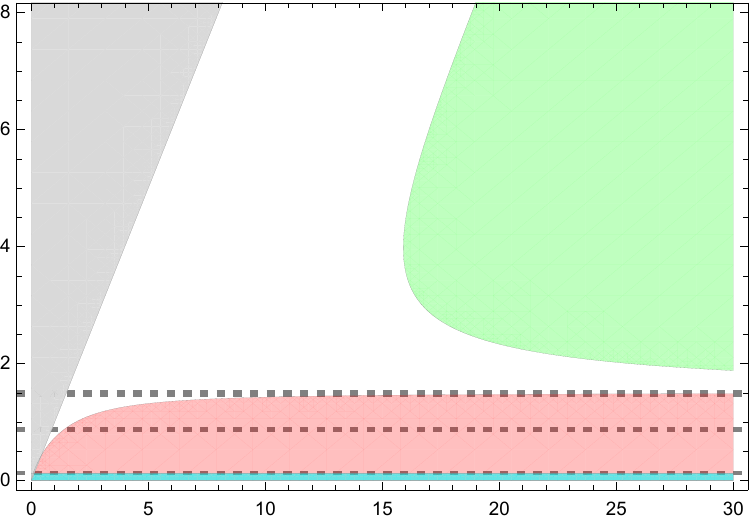}\label{fig:gb3}}
\\
\subfloat[]{\includegraphics[width=.45\textwidth]{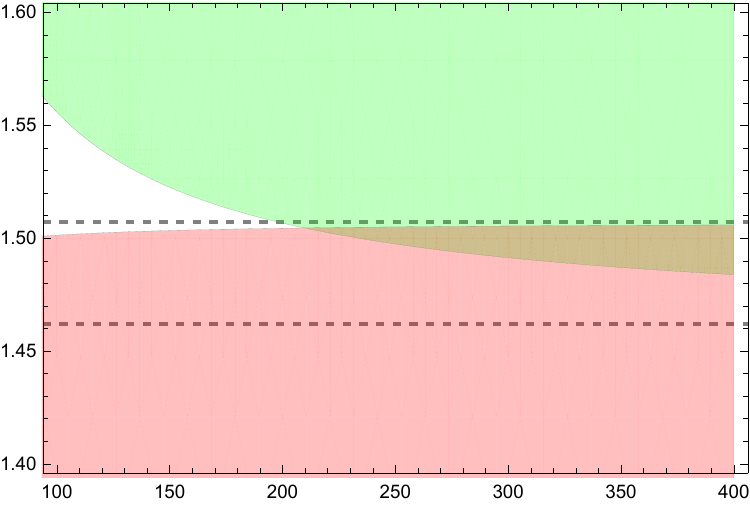}\label{fig:gb4}}
\qquad
\subfloat[]{\includegraphics[width=.45\textwidth]{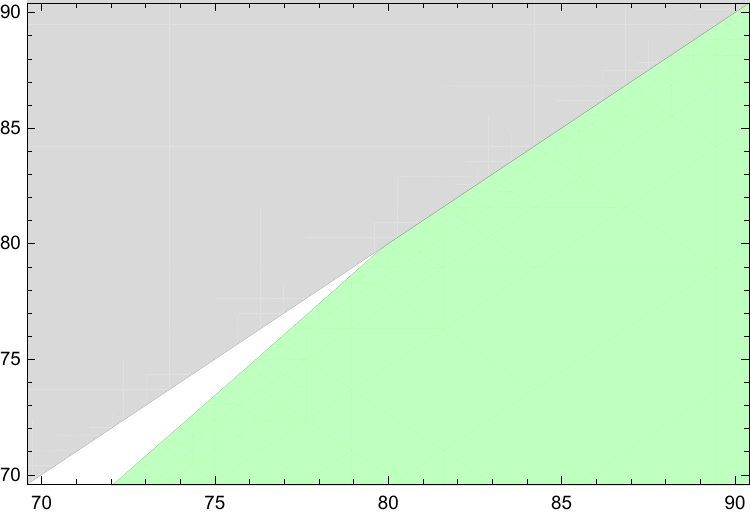}\label{fig:gb5}}
\caption{Figures \ref{fig:gb1}-\ref{fig:gb5} show different regions of the configuration space $(h_1,h_2) \in \mathbb{Z}_{+} \times \mathbb{Z}_{+}$. Units of $e^{J/2}$ are used on both axes. The colored areas respectively refer to: the condition $h_1 > h_2$ (in grey); Lemma \ref{lemma:cI} (in light blue); Lemma \ref{lemma:cII} (in green); Lemma \ref{lemma:cIII} (in red). The dashed horizontal lines correspond to: $h_2 = c_{I}\, e^{J/2}$ ($c_{I} = 0.123\,...$, see \eqref{eq:lemmacIb}), separating the blue and red regions; $h_2 = c_{II}(0)\, e^{J/2}$ ($c_{II}(0) = 1.461\,...$, see \eqref{eq:lemmacIIb}), approached asymptotically from above by the boundary of the green region for $h_1 \to +\infty$; $h_2 = c_{III}(0)\, e^{J/2}$ ($c_{III}(0) = 1.507\,...$, see \eqref{eq:lemmacIIIb}), approached asymptotically from below by the boundary of the red region for $h_1 \to +\infty$.}
\label{fig_gb}
\end{figure}

\subsubsection{Excluding tiles of finite size and bounded aspect ratio}

Lemmas \ref{lemma:cI}-\ref{lemma:cIII} imply that, for $J$ large, if $h_1<\infty$ and $(h_1,h_2)$ belongs to the union of the three regions identified by: 
\eqref{eq:lemmacIa}, the union over $\delta\in(0,1]$ of \eqref{eq:lemmacIIa}, and the union over $\delta\in(0,1]$ of \eqref{anot.1}, then $(h_1,h_2)$ is \textit{not}
a minimizer of $\mathcal E(h_1,h_2)$. In order to visualize these regions, see Fig. \ref{fig_gb}. The complement, i.e., the white region in Fig. \ref{fig_gb}, 
consists of pairs $(h_1,h_2)$ such that: $h_2/h_1$ is positive, uniformly in $J$, and smaller than $1$;  $h_2 e^{-J/2}$ is bounded from above and positive, uniformly in $J$. 
In particular, this white region is contained in 
$$\mathcal R:=\Big\{(h_1,h_2)\in \mathbb N^2 : (h_1,h_2)=\Big(\frac{h}\lambda, \frac{h}{1-\lambda}\Big),\ h\in e^{J/2}[c_{\min},c_{\max}],\ \lambda\in[\lambda_{\min},1/2]\Big\},$$
where $c_{\min},c_{\max},\lambda_{\min}$ are suitable positive constants, which can be chosen (sub-optimally) to be 
\begin{equation}\label{consts} 
c_{\min} = \frac{c_{III}(1)}2 = {0.356}\,..., \quad c_{\max} = c_{II}(1) = 79.819\,..., \quad \lambda_{\min} = \frac{\delta_{*}}{1+\delta_*} = 0.007\,...
\end{equation}
(here $\delta_{*}$ is determined by the condition $c_{II}(\delta_{*}) = c_{III}(\delta_{*})$). Therefore, in order for $(h_1,h_2)$ to be a minimizer of $\mathcal E(h_1,h_2)$ with $h_1\ge h_2$, either $(h_1,h_2)\in \mathcal R$, or $h_1=\infty$ (in which case, as discussed above, $h_2=h^*$). 
The following lemma excludes the possibility that $\mathcal R$ contains minimizers of $\mathcal E(h_1,h_2)$, 
thus concluding the proof that the only minimizers of $\mathcal E(h_1,h_2)$ with $h_1\ge h_2$ is $(\infty, h^*)$, as stated in our main theorem. 

\begin{lemma}\label{lemma:int}
There exists a (large, compared to $1$) positive constant $J_{\min}$ such that, if 
\begin{equation}\label{1.37}
J > J_{\min}\,, \quad c_{\min}\,e^{J/2} \leqslant h \leqslant c_{\max}\,e^{J/2}\,, \quad \lambda_{\min} \leqslant \lambda \leqslant 1/2\end{equation}
with $h/\lambda, h/(1-\lambda)\in\mathbb N$ and $c_{\min}, c_{\max}, \lambda_{\min}$ as in \eqref{consts}, then 
\begin{equation}
\mathcal{E}\Big({h \over \lambda}\,,\,{h \over 1 - \lambda}\Big)> \mathcal{E}_s(\lfloor h\rfloor).\end{equation}
\end{lemma}

\begin{proof}[Proof of Lemma \ref{lemma:int}]
Let $L$ be an integer divisible both by $h_1=h/\lambda$ and $h_2=h/(1-\lambda)$ with $h,\lambda$ as in \eqref{1.37}. We take $J_{\min}$ sufficiently large, so that conditions \eqref{1.37} imply that $h,h_1,h_2$ are also  large. First of all, 
by means of Eq. \eqref{eq:estripes}, 
we obtain
\begin{equation}
\begin{split}
& \mathcal{E}_s(\lfloor h\rfloor) - \mathcal{E}_s(h/\lambda) - \mathcal{E}_s\big(h/(1-\lambda)\big) \\
& \quad = - \,{4 \over h} \Big(\lambda \log \lambda + (1- \lambda) \log (1 - \lambda) \Big) + \mathcal{O}\big(h^{-2}\log h\big)\,. \end{split}\label{eq:DE2}
\end{equation}
Next, we compare the energy of the checkerboard phase $\bm{\sigma}_{c} \equiv \bm{\sigma}_{c}(h_1,h_2)$ with the sum of those of the auxiliary striped configurations 
$\bm{\sigma}_{V} \equiv \bm{\sigma}_{V}(h_1)$, consisting of vertical stripes of width $h_1$ and alternating spin signs, and $\bm{\sigma}_{H} \equiv \bm{\sigma}_{H}(h_2)$, 
consisting of horizontal stripes of width $h_2$ and alternating spin signs. We let $\Delta(\bm{\sigma}_{c})$ (resp. $\Delta^c(\bm{\sigma}_{c})$) be the union of positive (resp. negative) 
spin tiles of $\bm{\sigma}_c$ contained in $\Lambda_L$, and similarly for $\bm{\sigma}_V$ and $\bm{\sigma}_H$. One has $\Delta(\bm{\sigma}_{c}) = [\Delta(\bm{\sigma}_{V}) \cap \Delta(\bm{\sigma}_{H})] \cup [\Delta^{c}(\bm{\sigma}_{V}) \cap \Delta^{c}(\bm{\sigma}_{H})]$ and $\Delta^{c}(\bm{\sigma}_{c}) = [\Delta(\bm{\sigma}_{V}) \cap \Delta^{c}(\bm{\sigma}_{H})] \cup [\Delta^{c}(\bm{\sigma}_{V}) \cap \Delta(\bm{\sigma}_{H})]$. In view of these identities and of the fact that $\mathcal E(h_1,h_2)=L^{-2}\mathcal H_L(\bm{\sigma}_c)$, 
$\mathcal E_s(h_1)=L^{-2}\mathcal H_L(\bm{\sigma}_V)$ and $\mathcal E_s(h_2)=L^{-2}\mathcal H_L(\bm{\sigma}_H)$, we get
\begin{align}
& \mathcal{E}(h_1,h_2)-\mathcal{E}_s(h_1)- \mathcal{E}_s(h_2)
	\nonumber \\
& = -\, \frac2{L^2}\Bigg(\sum_{\mbox{{\scriptsize $\begin{matrix} \bm{x} \in \Delta(\bm{\sigma}_{c}) \vspace{-0.07cm}\\  \bm{y} \in \Delta^{c}(\bm{\sigma}_{c}) \end{matrix}$}}} - \sum_{\mbox{{\scriptsize $\begin{matrix} \bm{x} \in \Delta(\bm{\sigma}_{V}) \vspace{-0.07cm}\\  \bm{y} \in \Delta^{c}(\bm{\sigma}_{V}) \end{matrix}$}}} - \sum_{\mbox{{\scriptsize $\begin{matrix} \bm{x} \in \Delta(\bm{\sigma}_{H}) \vspace{-0.07cm}\\  \bm{y} \in \Delta^{c}(\bm{\sigma}_{H}) \end{matrix}$}}} \Bigg)\sum_{\bm{m} \in \mathbb{Z}^2} {1 \over |\bm{x} - \bm{y} + L \bm{m}|^3}
	\nonumber \\
	\nonumber \\
& = \frac8{L^2}\!\! \sum_{\mbox{{\scriptsize $\begin{matrix} \bm{x} \in \Delta(\bm{\sigma}_{V}) \cap \Delta(\bm{\sigma}_{H}) \vspace{-0.07cm}\\ \bm{y} \in \Delta^{c}(\bm{\sigma}_{V}) \cap \Delta^{c}(\bm{\sigma}_{H}) \end{matrix}$}}} \sum_{\bm{m} \in \mathbb{Z}^2} {1 \over |\bm{x} - \bm{y} + L \bm{m}|^3}\,. \label{eq:DE1}
\end{align}
To proceed, notice that the set $\Delta(\bm{\sigma}_{V}) \cap \Delta(\bm{\sigma}_{H})$ consists of $L/(2h_1) \times L/(2h_2) = \lambda(1-\lambda) L^2/(4 h^2)$ tiles. Taking this into account, 
and proceeding in a way similar to, but much simpler than, the one described in the Appendix (details left to the reader), we deduce
\begin{align}\label{daje}
\mathcal{E}(h_1,h_2) - \mathcal{E}_s(h_1)-\mathcal E_s(h_2)
> {8 \lambda(1-\lambda) \over h^2} \bigg(\,\sum_{\bm{x} \in T} \sum_{\bm{y} \in P} \,+\, {1 \over 2} \sum_{\bm{x} \in T} \sum_{\bm{y} \in \Xi}\,\bigg)\, {1 \over |\bm{x} - \bm{y}|^3}\,.
\end{align}
where: $T$ is any fixed rectangular tile in $\Delta(\bm{\sigma}_{V}) \cap \Delta(\bm{\sigma}_{H})$; $P$ is a tile identical to $T$, touching the latter in one of its vertexes; $\Xi$ is the half-stripe of width $h_1$ aligned with one of the short sides of $T$ and shifted upwards by $2h_2$ (see Fig. \ref{fig:ETPXi}).
\begin{figure}[t!]
\includegraphics[width=.45\textwidth]{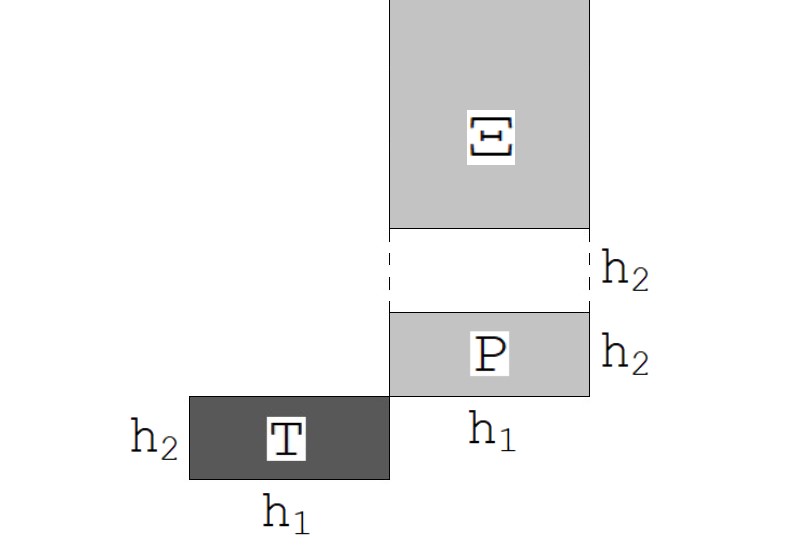}
\caption{The regions $T$, $P$ and $\Xi$ considered in the proof of Lemma \ref{lemma:int}.}\label{fig:ETPXi}
\end{figure}
Now, in order to evaluate the right side of \eqref{daje}, first of all note that 
\begin{equation}\label{pat.0}
{1 \over h}\sum_{\bm{x} \in T} \sum_{\bm{y} \in P}{1 \over |\bm{x} - \bm{y}|^3} = {1 \over h} \sum_{n_1\,=\,1}^{2h_1}\sum_{n_2\,=\,1}^{2h_2}\frac{ \min\{n_1,2h_1 - n_1\} \min\{ n_2 , 2h_2 - n_2\} }{(n_1^2 + n_2^2)^{3/2}} ,
\end{equation}
which is, for $h$ large, a Riemann sum approximation of 
\begin{equation}\label{pat.1}
\int_0^{2/\lambda} dx_1\int_0^{2/(1-\lambda)} dx_2 \frac{ \min\{x_1,2/\lambda - x_1\} \min\{ x_2 , 2/(1-\lambda) - x_2\} }{(x_1^2 + x_2^2)^{3/2}}.
\end{equation}
A patient evaluation of this integral gives
\begin{align*}
\eqref{pat.1}& =  {2 \over 1- \lambda} \Bigg[\log 2 + {2 \over \zeta}  \left( 2\, \sqrt{1 + \zeta^2/4} + \sqrt{1 + 4 \zeta^2} - 3\, \sqrt{1 + \zeta^2} \right) \\
& \qquad + {1 \over \zeta} \log\left({\big(1 + \sqrt{1 + \zeta^2}\big)^{3} \over \big(1 + \sqrt{1 + \zeta^2/4}\big)^{2} \big(1 + \sqrt{1 + 4 \zeta^2} \big)}\right) \\
& \qquad + \log\left({\big(\zeta + \sqrt{1 + \zeta^2}\big)^3 \over \big(2\zeta + \sqrt{1 + 4 \zeta^2}\big)^2 \big(\zeta/2 + \sqrt{1 + \zeta^2/4}\big)}\right)\Bigg]_{\zeta \,=\, \lambda/(1 - \lambda)}=:f(\lambda).\end{align*}
One can check that $f''(\lambda)$ is strictly negative in $[0,1/2]$ and, therefore, in this interval, $f(\lambda)$ is bounded from below by $f(0)+2\lambda(f(1/2)-f(0))=2\log 2+(0.3998\ldots)\lambda
\geqslant 2 \log 2 + \lambda/3$. Combining this with an estimate of the difference between \eqref{pat.1} and its Riemann approximation \eqref{pat.0} leads to 
\begin{equation}
{1 \over h}\sum_{\bm{x} \in T} \sum_{\bm{y} \in P}{1 \over |\bm{x} - \bm{y}|^3} \ge 2 \log 2 + \lambda/3+\mathcal O(h^{-1}\log h). 
\end{equation}
The second double sum in the right side of \eqref{daje} is bounded similarly: first, note that 
\begin{equation}\label{pat.03} {1 \over h} \sum_{\bm{x} \in T} \sum_{\bm{y} \in \Xi} {1 \over |\bm{x} - \bm{y}|^3} = {1 \over h} \sum_{n_1\,=\,1}^{2h_1}\sum_{n_2\,=\,2h_2 + 1}^{\infty}  \frac{\min\{n_1,2h_1 - n_1\}\min\{n_2 - 2h_2, h_2\}}{(n_1^2 + n_2^2)^{3/2}} \end{equation}
which is, for $h$ large, a Riemann sum approximation of 
\begin{equation}\label{pat.3}
\int_0^{2/\lambda} dx_1\int_{2/(1-\lambda)}^\infty dx_2 \frac{ \min\{x_1,2/\lambda - x_1\} \min\{ x_2 - 2/(1-\lambda),  1/(1-\lambda)\} }{(x_1^2 + x_2^2)^{3/2}}. 
\end{equation}
A patient evaluation of this integral gives
\begin{align*}
\eqref{pat.3}& = {1 \over 1 -\lambda} \Bigg[ - \log \zeta + 2 - 3 \log 3 + {4 \over \zeta} \left( \sqrt{1 + 4\zeta^2} + \sqrt{1 + 9\zeta^2/4} - \sqrt{1 + \zeta^2} - \sqrt{1 + 9\zeta^2}\right) \\
& \qquad + {2 \over \zeta} \log\left({\big(1 + \sqrt{1 + \zeta^2}\big)\big(1 + \sqrt{1 + 9\zeta^2}\big) \over \big(1 + \sqrt{1 + 4\zeta^2}\big)\big(1 + \sqrt{1 + 9\zeta^2/4}\big)}\right) \\
& \qquad + \log\left({\big(\zeta + \sqrt{1 + \zeta^2}\big)^2\big(3 \zeta + \sqrt{1 + 9\zeta^2}\big)^6 \over \big(2 \zeta + \sqrt{1 + 4\zeta^2}\big)^4 \big(3 \zeta/2 + \sqrt{1 + 9\zeta^2/4}\big)^3}\right)
\Bigg]_{\zeta\,=\,\lambda/(1-\lambda)}=: -\log\lambda+g(\lambda).
\end{align*}
One can check that the second derivative of the function $g(\lambda)$ defined here is strictly negative in $[0,1/2]$ and, therefore, in this interval, $g(\lambda)$ is bounded from below by $g(0)+2\lambda(g(1/2)-g(0))=
2 - 3 \log 3 + (1.497\ldots)\lambda
\geqslant 2 - 3 \log 3 + \lambda$. Combining this with an estimate of the difference between \eqref{pat.3} and its Riemann approximation \eqref{pat.03} leads to 
\begin{equation} {1 \over h} \sum_{\bm{x} \in T} \sum_{\bm{y} \in \Xi} {1 \over |\bm{x} - \bm{y}|^3} \ge -\log\lambda + 2 - 3 \log 3 + \lambda+\mathcal O(h^{-1}).\end{equation}
Putting things together, the previous estimates imply 
\begin{align*}
& \mathcal{E}(h_1,h_2)-\mathcal E_s(\lfloor h\rfloor) > {4 \over h} \bigg[ \left(2 - \log\Big({27 \over 16}\Big)\right) \lambda + \left(\log\Big({27 \over 16}\Big) - {1 \over 3}\right)\lambda^2 - {5 \over 3}\,\lambda^3 \\
&\hspace{4.5cm} + \lambda^2 \log \lambda + (1- \lambda) \log (1 - \lambda) \bigg] \,+\,\mathcal{O}\big(h^{-2} \log h\big)  \,,
\end{align*}
which ultimately yields the thesis, since the function between square brackets on the right-hand side is non-negative for all $\lambda\in[0,1/2]$ and only vanishes for $\lambda \to 0^{+}$.
\end{proof}

\section*{Appendix. Proofs of \eqref{weare} and \eqref{appp}.}

In order to prove \eqref{weare}, we let 
\begin{align*}
R_{II} & := \sum_{\bm{x} \in T} \Big( \sum_{\bm{y} \in \Delta_{e}^{+} \cup S_{1}^{+} \cup S_{4}^{+}} - \sum_{\bm{y} \in \Delta_{e}^{-} \cup S_{1}^{-} \cup S_{4}^{-}} \Big) \sum_{\bm{m} \in \mathbb{Z}^2} {1 \over |\bm{x} - \bm{y} + L \bm{m}|^3} \\
& \hspace{4cm} - \sum_{\bm{x} \in T} \bigg( \sum_{\bm{y} \in S_1} - \;2 \sum_{\bm{y} \in S_4} - \;4 \sum_{\bm{y} \in \Xi} \bigg)\, {1 \over |\bm{x} - \bm{y}|^3} 
\end{align*}
and note that this can be equivalently rewritten as 
\begin{equation}R_{II} = \sum_{\bm{x} \in T} \sum_{T' \in \mathcal{R}_{II}} \sum_{\bm{y} \in T'} {\sigma(T') \over |\bm{x} - \bm{y}|^3} ,\end{equation}
where the set $\mathcal{R}_{II}$ and the spins $\sigma(T') \in \{\pm 1, + 2\}$ of the tiles $T'$ forming it are described in Fig. \ref{fig:AP2RII}. 
\begin{figure}[t!]\label{fig:AP2RII}
\includegraphics[width=.6\textwidth]{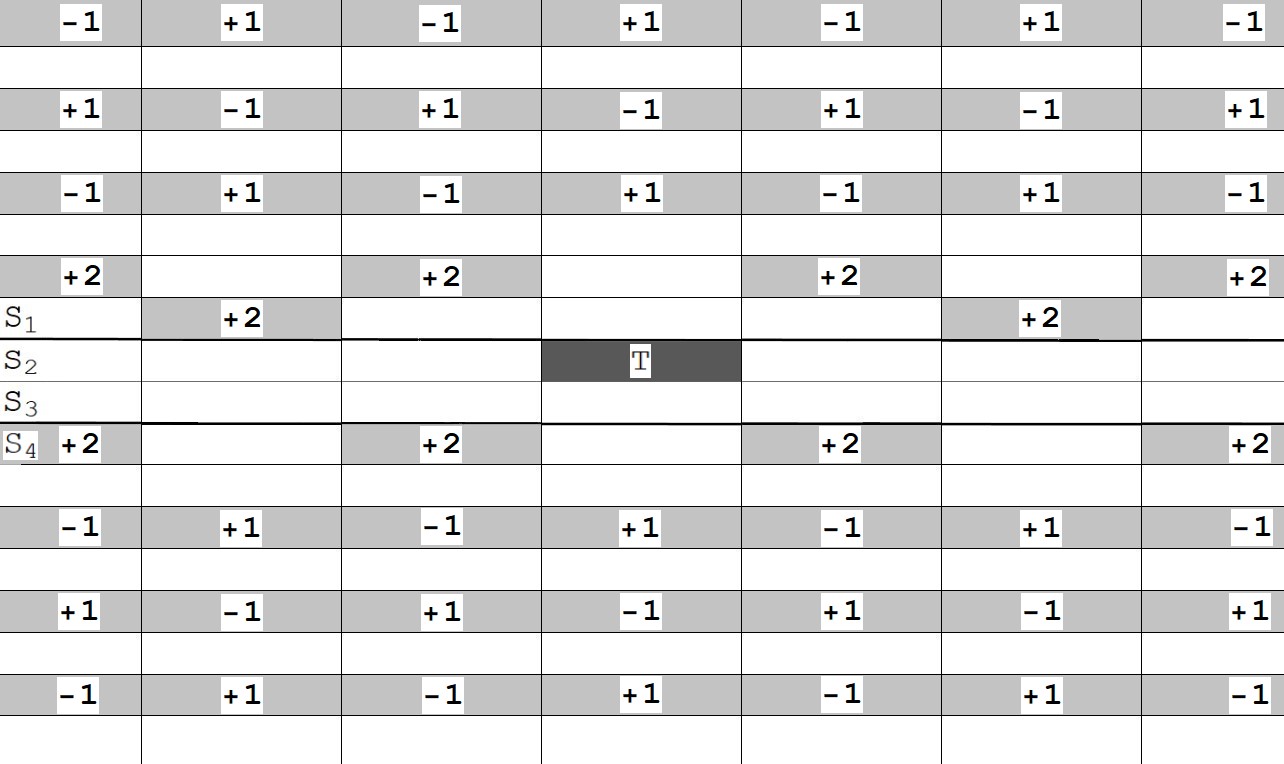}
\caption{The grey regions represent the set $\mathcal{R}_{II}$ mentioned in the proof of Lemma \ref{lemma:cII}. The signed numbers indicate the uniform spin $\sigma(T')$ of the tile $T'$ belonging to $\mathcal{R}_{II}$.}\label{fig:AP2RII}
\end{figure}
Then, by simple translation and symmetry arguments, using the monotonicity of $ {1/ |\bm{x} - \bm{y}|^3}$ we readily get
\begin{equation*}
R_{II} > 0\,,
\end{equation*}
as desired. In order to prove \eqref{appp} we proceed similarly. Let
\begin{align*}
R_{III} & := 2 \bigg(\sum_{\bm{x} \in T_{a}} \sum_{\bm{y} \in \Pi} + \sum_{\bm{x} \in T_{b}} \sum_{\bm{y} \in \Xi} \,-\, 4\! \sum_{\bm{x} \in T_{a}} \sum_{\bm{y} \in P} \,-\, 2\! \sum_{\bm{x} \in T_{a}} \sum_{\bm{y} \in Q}\bigg)\, {1 \over |\bm{x} - \bm{y}|^3} \\
& \hspace{4cm} - \sum_{\bm{x} \in T}\bigg( \sum_{\bm{y} \in \Delta_{e}^{-}} - \sum_{\bm{y} \in \Delta_{e}^{+}} \bigg) \sum_{\bm{m} \in \mathbb{Z}^2} {1 \over |\bm{x} - \bm{y} + L \bm{m}|^3} 
\end{align*}
and note that this can be equivalently rewritten as 
\begin{equation}
R_{III} = 2 \sum_{\bm{x} \in T_{a}} \sum_{T' \in \mathcal{R}_{III,a}} \sum_{\bm{y} \in T'} {\sigma(T') \over |\bm{x} - \bm{y}|^3}
	+ 2 \sum_{\bm{x} \in T_{b}} \sum_{T' \in \mathcal{R}_{III,b}} \sum_{\bm{y} \in T'} {\sigma(T') \over |\bm{x} - \bm{y}|^3} \,,\end{equation}
where the sets $\mathcal{R}_{III,a},\mathcal{R}_{III,b}$ and the spins $\sigma(T') \in \{\pm 1, + 2\}$ of the tiles $T'$ forming them are described in Figs. \ref{fig:AP3RIIIa} and \ref{fig:AP3RIIIb}. 
\begin{figure}[t!]
\subfloat[The grey regions represent the set $\mathcal{R}_{III,a}$ mentioned in the proof of Lemma \ref{lemma:cIII}. The signed numbers indicate the uniform spin $\sigma(T')$ of the tile $T'$ belonging to $\mathcal{R}_{III,a}$.]{\includegraphics[width=.45\textwidth]{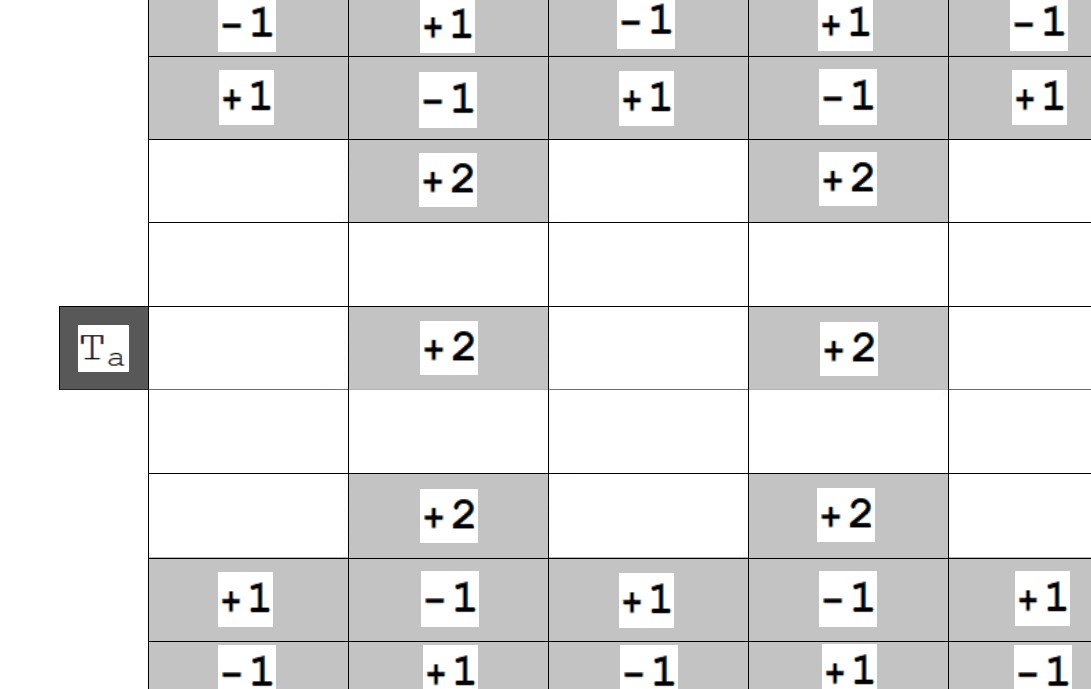}\label{fig:AP3RIIIa}}
\qquad
\subfloat[The grey regions represent the set $\mathcal{R}_{III,b}$ mentioned in the proof of Lemma \ref{lemma:cIII}. The signed numbers indicate the uniform spin $\sigma(T')$ of the tile $T'$ belonging to $\mathcal{R}_{III,b}$.]{\includegraphics[width=.45\textwidth]{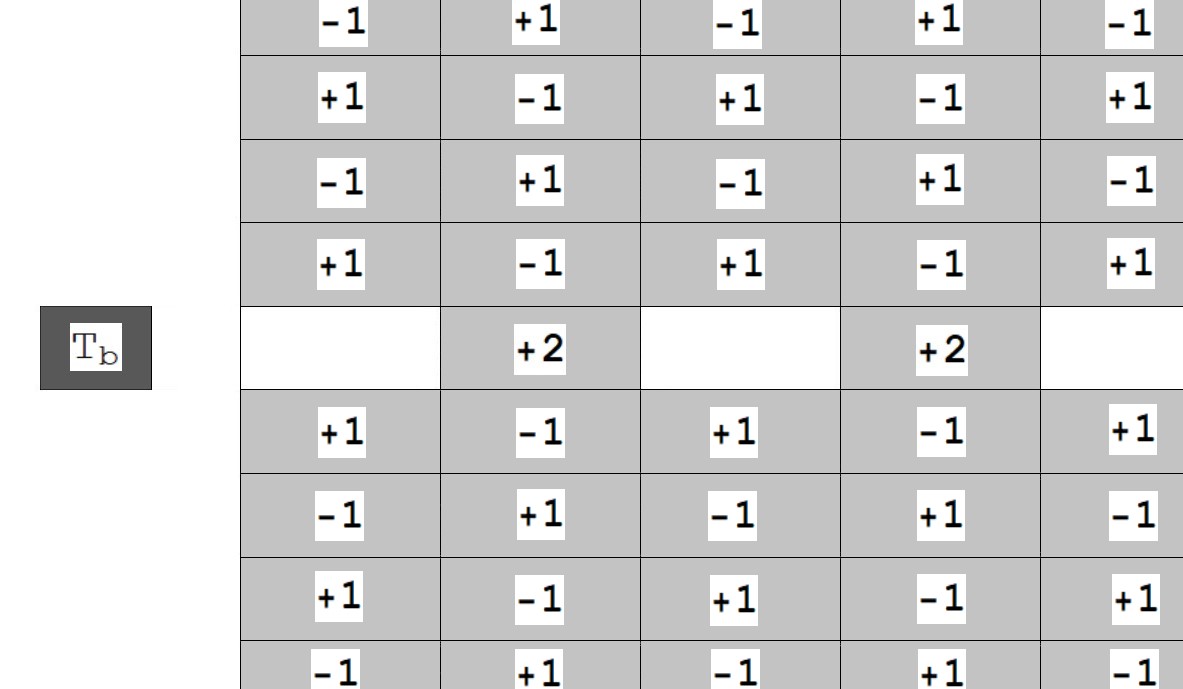}\label{fig:AP3RIIIb}}
\caption{}
\end{figure}
Then, by simple translation and symmetry arguments, using the monotonicity of $ {1/ |\bm{x} - \bm{y}|^3}$ we readily get
\begin{equation*}
R_{III} > 0\,,
\end{equation*}
as desired. 

%------
% Insert acknowledgments and information
% regarding funding at the end of the last
% section, i.e., right before the bibliography.
%------

\begin{ack}
We thank Elliott and Joel for proposing to one of us this project and for the countless exciting discussions we had on stripes and periodic states over the years. We wish Elliott many many more years of 
enthusiastic scientific activity and happy life. [Some of the results presented in this paper have been derived using the software \textsc{Mathematica} for both symbolic and numerical computations.]
\end{ack}

\begin{funding}
This work has been supported by: the European Research Council (ERC) under the European Union’s Horizon 2020 research and innovation programme (ERC CoG UniCoSM, grant agreement n. 724939);
MIUR, PRIN 2017 project MaQuMA, PRIN201719VMAST01; INdAM-GNFM Progetto Giovani 2020 ``\textsl{Emergent Features in Quantum Bosonic Theories and Semiclassical Analysis}''.
\end{funding}

%------
% Insert the bibliography.
%------

\end{document}